\documentclass[11pt,a4paper]{article}
\usepackage[margin=2.2cm]{geometry}
\usepackage{amsmath, amssymb, amsthm}
\usepackage{booktabs, array}
\usepackage{graphicx}
\usepackage{float}
\usepackage[authoryear,round]{natbib}
\usepackage[hidelinks]{hyperref}
\usepackage[T1]{fontenc}
\usepackage{microtype}
\usepackage{enumitem}
\usepackage{needspace}
\newtheorem{theorem}{Theorem}

\newtheorem{lemma}[theorem]{Lemma}
\newtheorem{corollary}[theorem]{Corollary}
\theoremstyle{definition}

\theoremstyle{remark}
\newtheorem{remark}{Remark}

\newcommand{\RR}{\mathbb{R}}
\newcommand{\dd}{\,d}
\newcommand{\one}{\mathbf 1}
\newcommand{\eee}[1]{\times 10^{-#1}}

\newif\ifblind
\blindfalse

\title{Closed forms all the way up:\\
the sample variance of four observations}
\ifblind
\author{Anonymized for review}
\else
\author{%
  Remus Osan\thanks{Gilead Sciences, Foster City, CA 94404, USA.
    \texttt{remus.osan@gilead.com}.}%
}
\fi
\date{October 10, 2026}

\begin{document}
\maketitle

\begin{abstract}
The distribution function $F(y)=\Pr(S^2\le y)$ of the sample variance $S^2$ of
four independent observations is elementary on its entire support for every polynomial parent
density on $[0,1]$ (without loss of generality), and we give it in finite closed form for the uniform, Beta$(2,1)$ and
Beta$(2,2)$ parents.  The support splits at $y = 1/6$, $2/9$ and $1/4$, with the ceiling at
$1/3$, all fixed by geometry: in coordinates aligned with the cube diagonal the
event $S^2\le y$ is a ball of radius $\sqrt{3y}$, the cube projects to a rhombic dodecahedron,
and these are the levels at which the ball reaches its faces, edges and two kinds of
vertices.  The divergence theorem reduces every moment of the clipped ball to integrals
along circular arcs, so each regime closes in algebraic functions of $\sqrt{y}$ and
inverse trigonometric functions.  For the uniform parent the only non-obvious constant is
$K_g = 12\sqrt2\arctan\sqrt2$, proportional to the mean width of the cuboctahedron; the Beta
constants are rational expressions in $\pi$, $\sqrt2$, $\sqrt3$ and $\arctan\sqrt2$, as in the
uniform case.  At $n = 4$ the
normal-theory $\chi^2$ recipe used in acceptance testing misplaces the $99$th percentile of
$S^2$ by $50\%$ for the uniform parent.  For three observations, polynomial parents also give
elementary closed forms (companion paper), but two elementary cases in a row do not
guarantee a third, and we make no claim beyond $n = 4$.
\end{abstract}



\section{Introduction}\label{sec:intro}

The distribution of the sample variance of a handful of observations
from a \emph{bounded} parent is a problem with an old history and a
live industrial application.  \citet{rietz1931} obtained the
$n = 3$ uniform law in 1931, as the density of the standard
deviation.\footnote{The result originated as a problem in solid
geometry---the area of a diagonal cylinder's surface lying inside a
unit cube---posed by P.~R. Rider (Problem~3413, \textit{Amer.\ Math.\
Monthly} \textbf{37} (1930), 157) and solved by Rietz (\textit{ibid.}\
\textbf{38} (1931), 292--295); Rietz reinterpreted the same integral
as a probability \citep{rietz1931}.}  Rietz also conjectured that the
density of $S$ is a polynomial of degree $n-1$ for small $s$; \citet{aliRichards1975} proved it by deriving, for
every $n \ge 3$, an elementary formula for the lowest
regime---in our normalization $Y \le \tfrac{1}{2(n-1)}$---of the
form $c_1 r^{n-1} - b(n)\,r^{n}$ with $r = \sqrt{(n-1)Y}$,
evaluating the constant $b(n)$ exactly for $n \le 5$ and noting
that beyond that regime the shaved volume ``is generally difficult
to compute''; to our knowledge, no finite closed form beyond the
first regime has been published at any $n \ge 4$.
\citet{royen2007b,royen2008} gave an infinite Fourier series
valid for every bounded continuous parent and every $n$, without
offering a closed form beyond the first regime, and later work
resorts to approximations \citep{provost2013}.  Content-uniformity testing in pharmaceutical
manufacturing \citep{usp905,williamsHauck2002,bergum1990} evaluates
acceptance criteria built on $(\bar X, s)$ from small samples of a
quantity that is physically confined to a bounded range and typically
non-Gaussian \citep{orr1978}: the content of a unit is set by the
number and size of the drug-substance particles it captures, bulk
and particle properties that are characterized by imaging and
related methods \citep{morrison2021,tin2025,mohan2026} and that propagate to the
unit-content distribution through stochastic tablet-formation
models \citep{chu2025}.  Exact small-$n$ variance laws for bounded
parents are the natural computational currency there, and the
computation of the acceptance probabilities themselves from such
laws, for the mean--spread rules the standards use, is carried out
in \citet{enginecompanion}.    Nor are samples of three
or four hypothetical in acceptance practice: the variables-sampling
standard MIL-STD-414 \citep{milstd414}, from which ANSI/ASQ~Z1.9
descends, assigns sample size code letters $B$ and $C$ --- sample
sizes $3$ and $4$ --- to small lots, and its standard-deviation method
accepts when $(U - \bar X)/s \ge k$, a criterion in the sample mean and
the sample standard deviation (Table~B-1; $k = 1.12$ at $n = 3$ and
$1.17$ at $n = 4$ for an acceptable quality level of $2.5\%$).
A companion paper \citep{n3paper} treats
$n = 3$ in this spirit and maps which parent densities admit finite
closed forms.

This paper settles the four-observation case for the uniform parent
on the complete support.
Let $X_1,\dots,X_4$ be i.i.d.\ uniform on $[0,1]$ and
\begin{equation}
  S^2 \;=\; \tfrac13 \sum_{i=1}^4 (X_i - \bar X)^2, \qquad
  F(Y) \;=\; \Pr(S^2 \le Y), \qquad Y \in [0, \tfrac13].
\end{equation}
Our main result, Theorem~\ref{thm:main}, is the exact $F$ in finite
elementary closed form on all of $[0, 1/3]$: four analytic pieces,
joined at $Y = 1/6$, $2/9$, and $1/4$, built from
$\sqrt{Y},\,Y,\,Y^{3/2},\,Y^2$, inverse trigonometric functions of
algebraic arguments, and the constant
\begin{equation}
  K_g \;=\; 12\sqrt{2}\,\arctan\sqrt{2} \;=\; 16.2122606\ldots,
  \label{eq:Kg-intro}
\end{equation}
which is $2\pi$ times the mean width of the cuboctahedron and plays
the role that $\pi$ plays at $n = 3$.

\paragraph{Why four.}
Two facts place $n = 4$.  The first is that it is the smallest
sample size at which the closure question is genuinely open.  In
every $n$ the reduction below writes the CDF as an integral over the
unit sphere $S^{n-2}$ of contrast space, with the region where the
cube constraint binds excised; the excised region is a union of caps
around the $n(n-1)$ root directions, and the chamber walls of the
braid arrangement cut those caps into pieces.  The objects that carry
the regime structure are the pieces of cap boundary inside a chamber,
which have dimension $n - 3$.  At $n = 3$ the sphere is a circle and
the pieces are points: the limits of integration are angles that
depend on $Y$ alone, and the CDF closes in inverse trigonometric
functions of algebraic arguments essentially by inspection.  At
$n = 4$ the pieces are arcs, and their length $\beta(\psi)$ is a
\emph{function} of the colatitude that must itself be integrated.
This is the first $n$ at which a limit of integration is a
transcendental function rather than a number, so elementarity is no
longer automatic.  Theorem~\ref{thm:main} shows that it holds for the
uniform parent, and Theorem~\ref{thm:poly} for every polynomial
parent, by a mechanism that uses the arcs: the divergence theorem
turns every moment of the clipped ball into trigonometric integrals
along circular arcs and areas of regions bounded by them, and both
are elementary.  The argument offers no guarantee for larger $n$,
where the pieces are polygons, and we make no claim there.  The second fact is that the exact law
is not made redundant by the normal limit.  For a parent with nonzero
excess kurtosis $\kappa$, the uniform among them, the recipe in
acceptance-testing use, $(n-1)S^2/\sigma^2 \sim \chi^2_{n-1}$,
misstates the variance of $S^2$ by a factor that does not tend to
one: $\operatorname{Var} S^2 = \sigma^4[2/(n-1) + \kappa/n]$, so the
actual variance divided by the recipe's tends to $1 + \kappa/2$,
which is $2/5$ for the uniform parent.  For that parent the recipe's
grid-maximum error against the exact CDF (Table~\ref{tab:approx}) grows from $0.064$ at $n = 3$ to
$0.102$ at $n = 6$, toward a limiting Kolmogorov discrepancy of
$0.109$, while
its relative error in the $99$th percentile vanishes only like
$1.21\,n^{-1/2}$ and is still $42\%$ at $n = 6$
(Section~\ref{sec:ledger}).

\paragraph{The geometry.}
The sample $(X_1, X_2, X_3, X_4)$ is a point of the unit $4$-cube
$[0,1]^4$, whose main diagonal is the segment
$\{t\,(1,1,1,1) : 0 \le t \le 1\}$ joining the corner $(0,0,0,0)$ to
the corner $(1,1,1,1)$; the sample mean places the point
$(\bar X, \bar X, \bar X, \bar X)$ on it.  Work in an orthonormal
basis split along that diagonal: write each observation as mean plus
deviation, $X_i = \bar X + w_i$.  The deviation vector $w$ lives in the
$3$-dimensional hyperplane $\{w : \sum_i w_i = 0\}$ orthogonal to
the diagonal---the
\emph{contrast space}, in the analysis-of-variance sense---and it
carries all of the sample variance: $\sum_i w_i^2 = 3S^2$.  In
words, $F(Y)$ is the fraction of the cube lying within distance
$\sqrt{3Y}$ of its main diagonal.  The
variance constraint is therefore a centered ball of radius
$R = \sqrt{3Y}$ in contrast space, and the unit cube projects to a
\emph{rhombic dodecahedron} \citep{coxeter1973}: the convex
polyhedron bounded by twelve congruent rhombic faces, with
twenty-four edges and fourteen vertices of two kinds---eight
\emph{corner} vertices, where three faces meet, and six \emph{axial}
vertices, where four meet.  This is the
three-dimensional ascent of the picture that organizes the companion
$n = 3$ paper \citep{n3paper}: there the variance set is a disk
growing inside the cube's hexagonal shadow, with a single interior
threshold at $Y = 1/4$, where the disk reaches the hexagon's edges,
and the ceiling $Y = 1/3$ at its vertices.  One dimension up, the
hexagon becomes the rhombic dodecahedron, and its richer
boundary---faces, edges, corner vertices, axial vertices---turns one
interior threshold into three.  The four thresholds are the four distances
from the center to the polytope's boundary orbits:
\[
  \underbrace{\tfrac{1}{\sqrt2}}_{\text{12 faces}} < \;
  \underbrace{\sqrt{\tfrac23}}_{\text{24 edges}} < \;
  \underbrace{\tfrac{\sqrt3}{2}}_{\text{8 corner vertices}} < \;
  \underbrace{1}_{\text{6 axial vertices}}
  \qquad\Longleftrightarrow\qquad
  Y \in \Bigl\{\tfrac16,\ \tfrac29,\ \tfrac14,\ \tfrac13\Bigr\}.
\]
Every qualitative feature of the CDF is visible in
Figure~\ref{fig:rd-view}: nothing happens until the ball reaches the
faces, caps are shaved off between faces and edges, the caps' bounding
wall arcs interact between edges and corner vertices, the arcs merge
between corner and axial vertices, and $F$ reaches $1$ exactly when
the ball swallows the polytope.

\begin{figure}[!t]
\centering
\begin{minipage}[c]{0.31\linewidth}
  \centering
  \includegraphics[width=\linewidth]{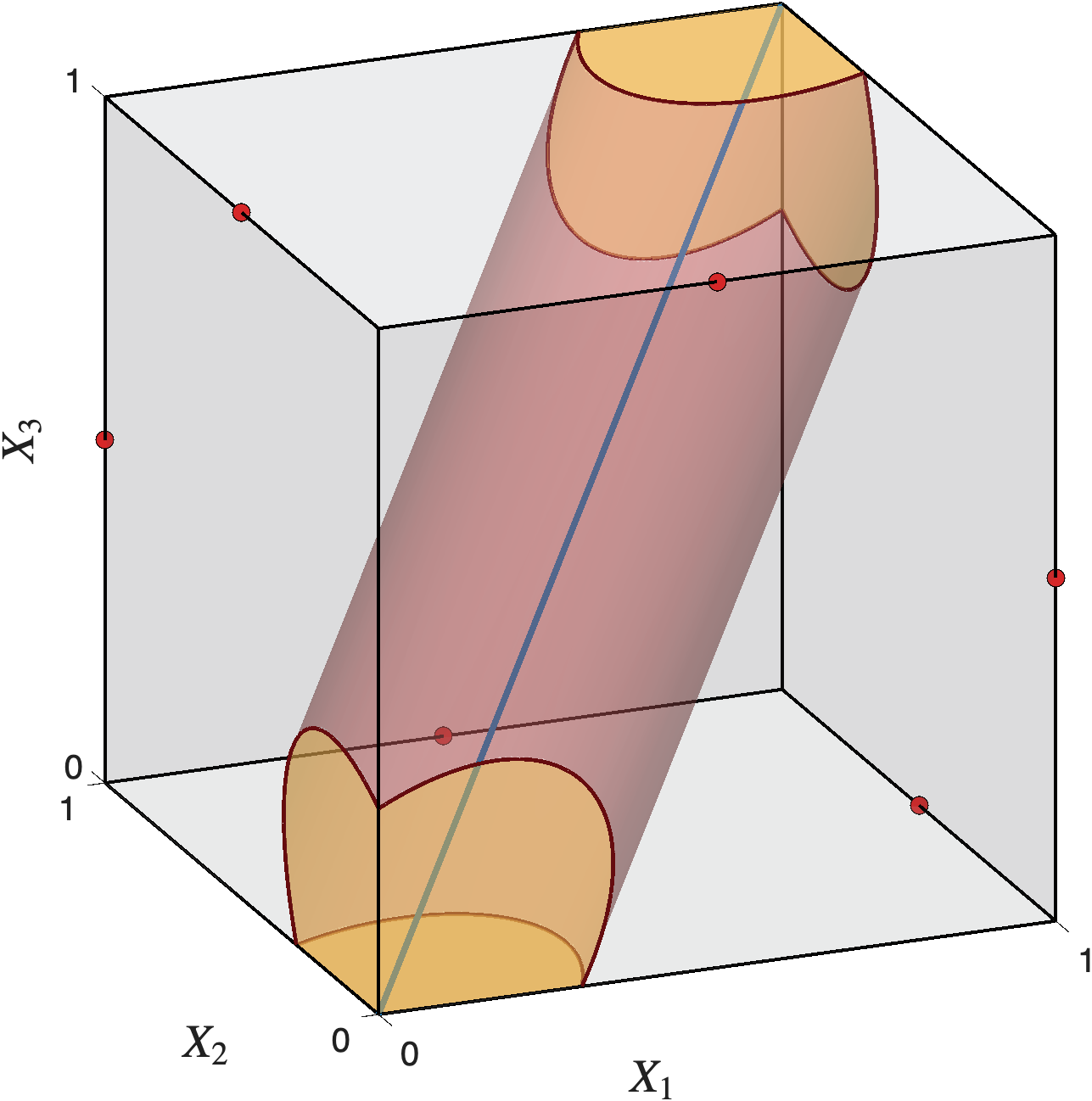}\\[6pt]
  \includegraphics[width=\linewidth]{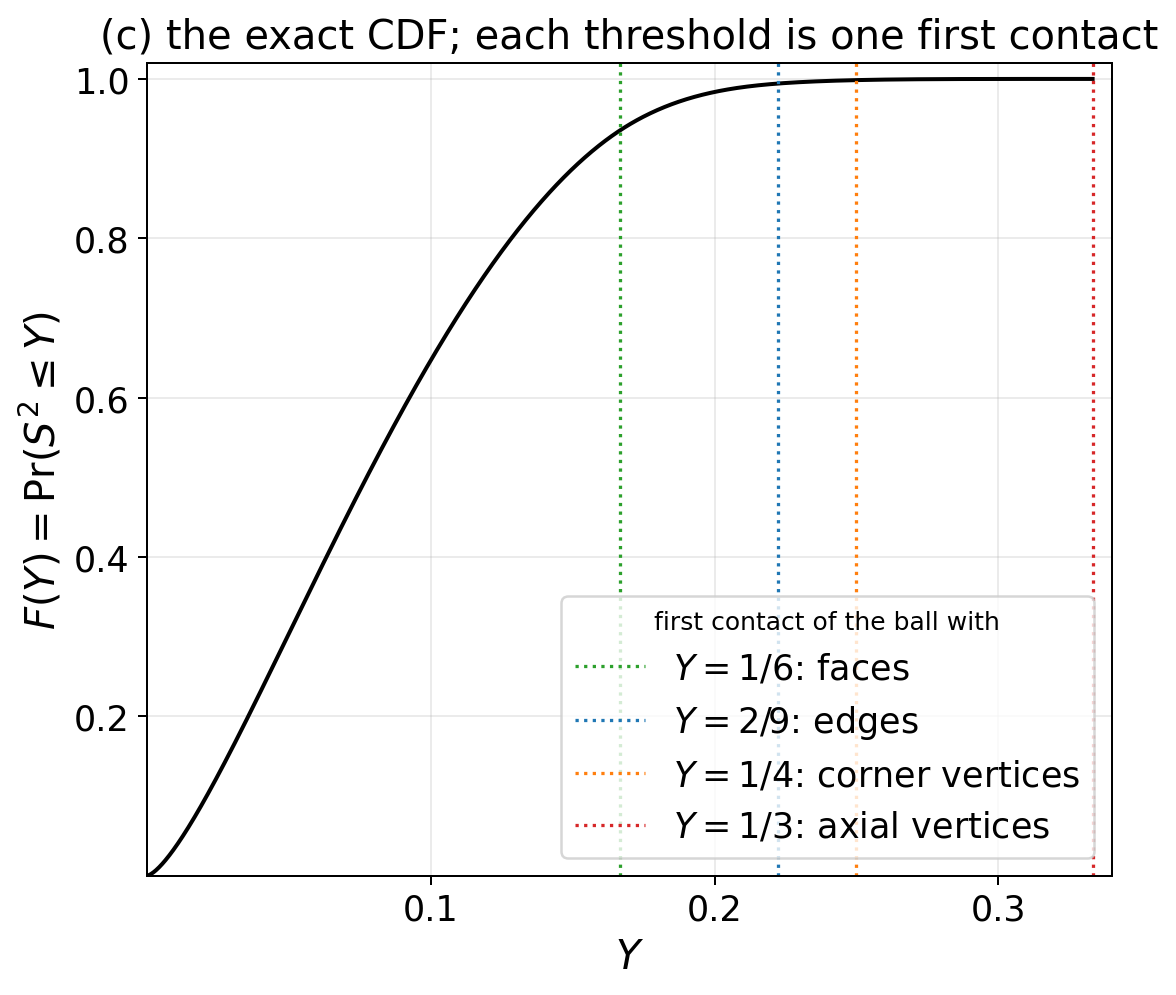}
\end{minipage}\hfill
\begin{minipage}[c]{0.66\linewidth}
  \centering
  \includegraphics[width=\linewidth]{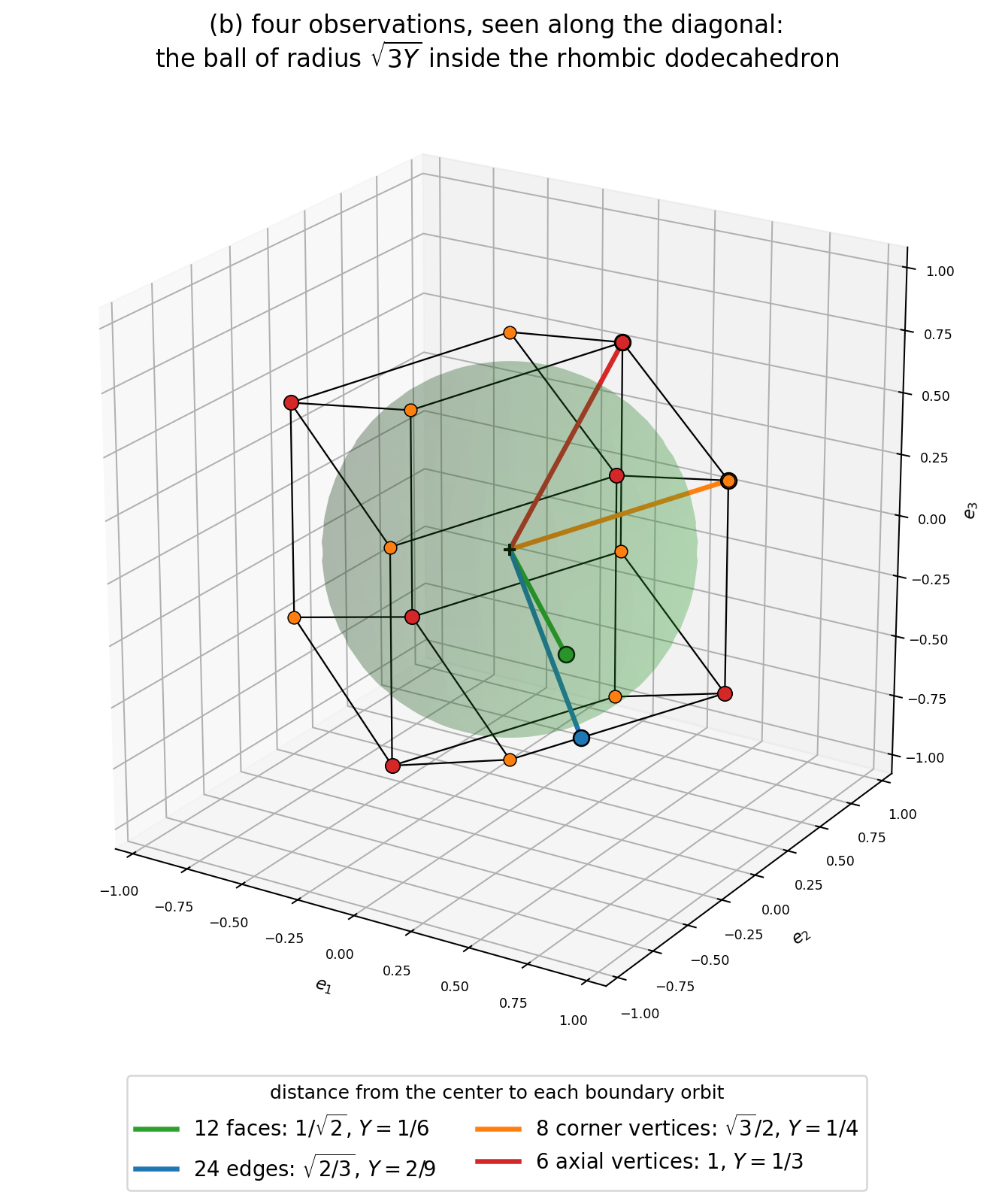}
\end{minipage}
\caption{What the CDF measures.  A sample is a point of the unit
cube, its mean is its position along the main diagonal, and its
squared distance from that diagonal is $n - 1$ times its variance,
so $F(Y)$ is the fraction of the cube lying within distance
$\sqrt{(n-1)Y}$ of the diagonal.  \emph{(a)} The picture one dimension down, where it can
be drawn, reproduced from the companion paper \citep{n3paper}: for
three observations the region is the part of the $3$-cube within
distance $\sqrt{2Y}$ of its diagonal (blue), a cylinder (pink wall)
shaved wherever it leaves the cube (orange cut faces), shown at
$Y = 0.03$, where it already meets the faces near the diagonal's
endpoints.  As $Y$ grows the cylinder first reaches the six cube
edges not incident to those endpoints, at their midpoints (red,
added here), at $Y = 1/4$, where the formula for $F$ changes, and
reaches the vertices at $Y = 1/3$, where $F = 1$.
\emph{(b)} For four observations, look along the diagonal: the
cube's shadow is the rhombic dodecahedron and the tube is a ball of
radius $\sqrt{3Y}$, drawn at the moment it touches the faces.  The
formula changes each time the growing ball first touches a new kind
of boundary point, and the four colored segments are those four
distances: the $12$ faces at $Y = 1/6$, the $24$ edges at
$Y = 2/9$, the $8$ corner vertices at $Y = 1/4$, and the $6$ axial
vertices at $Y = 1/3$, where the ball swallows the polytope and
$F = 1$.  \emph{(c)} The exact CDF with the four contacts marked;
nothing is visible to the eye, and the thresholds show up only in
the derivatives.}
\label{fig:rd-view}
\end{figure}

\paragraph{Contributions.}
\begin{enumerate}[leftmargin=1.6em, itemsep=1pt]
\item \textbf{The exact CDF} (Theorem~\ref{thm:main}): finite
  elementary closed form on all of $[0,1/3]$, with fully explicit
  antiderivatives in every regime and exact matching constants
  ($\pm\pi/\sqrt2$ at the first two joints; a closed-form ceiling
  anchor at the last).
\item \textbf{The geometric dictionary}
  (Section~\ref{sec:geometry}): thresholds $=$ boundary-orbit
  distances of the rhombic dodecahedron; the constant $K_g$ $=$
  $2\pi$ times the mean width of its polar dual, the cuboctahedron,
  via the $A_3$ root system.
\item \textbf{Validation at scale} (Section~\ref{sec:validation}):
  agreement with an independent deterministic quadrature to
  $1.3\eee{7}$ at every tabulated point, with Monte Carlo at
  $N = 2\times 10^7$, plus three internal exactness checks: density
  continuity at the thresholds, a symbolic zero for the explicit-form
  gap at $Y = 1/4$, and $F(1/3) - 1 = 5.7\eee{14}$ under numerical
  assembly.
\item \textbf{Polynomial parents} (Section~\ref{sec:poly}): the
  closure is inherited by every polynomial parent density, proved by
  reducing every polynomial moment of the clipped ball to boundary
  integrals; the four-regime exact CDFs are worked out in full for
  Beta$(2,1)$ and Beta$(2,2)$, with the regime joints at $Y = 1/6$ and
  the ceiling identity $F(1/3)=1$ certified symbolically, and the
  joints at $Y = 2/9$ and $Y = 1/4$ verified to $10^{-50}$ at $60$-digit
  working precision.
\item \textbf{Comparison with Royen's series and with the normal
  limit} (Sections~\ref{sec:royen} and~\ref{sec:ledger}): Royen's
  series, assembled with its exact constant, reaches $10^{-3}$ in
  about ten terms in the interior of the support but keeps an
  endpoint error near $Y = 0$, while the closed form is exact
  everywhere at fixed cost; and the exact CDFs at $n = 3$ to $6$ show
  the $\chi^2$ recipe's CDF error growing with $n$ while a
  kurtosis-corrected normal improves.
\end{enumerate}

\section{The diagonal-orthogonal coordinate system}\label{sec:coords}

Let $E$ be a $3\times4$ matrix with orthonormal rows spanning
$\one^\perp \subset \RR^4$; the Helmert contrasts are a standard
choice.  This route into the problem is itself classical: the proof
of \citet{aliRichards1975} already passes through
the ordered chamber, a Helmert rotation with one axis along the
diagonal, and polar coordinates on the contrast space---the
apparatus of this section, in 1975 notation.  What their treatment
did not extract, and what drives everything below, is the geometry
those coordinates expose.  Split $X \in \RR^4$ orthogonally as
\begin{equation}
  X = \mu\,\one + E^{\!\top} y, \qquad
  \mu = \bar X, \quad y = E X \in \RR^3 .
\end{equation}
Then $X_i = \mu + w_i$ with $w = E^{\!\top}y$, $\sum_i w_i = 0$, and
\begin{equation}
  3\,S^2 = \sum_i (X_i - \bar X)^2 = \|y\|^2,
  \qquad\text{so}\qquad
  \{S^2 \le Y\} = \{\|y\| \le R\}, \quad R = \sqrt{3Y}.
\end{equation}
The variance constraint is a centered \emph{ball} in contrast space.
The change of variables has $\dd X = 2\,\dd\mu\,\dd y$, and for the
uniform parent the $\mu$-integration is immediate: $\mu$ ranges over
the slab $[-\min_i w_i,\ 1 - \max_i w_i]$, of length
$\bigl(1 - g(w)\bigr)_+$, where
\begin{equation}
  g(w) \;=\; \max_i w_i - \min_i w_i
  \label{eq:range-gauge}
\end{equation}
is the \emph{range} of the contrast vector.  Hence the entire CDF is
one three-dimensional integral,
\begin{equation}
  F(Y) \;=\; 2 \int_{\|y\|\le R}
     \bigl(1 - g(E^{\!\top}y)\bigr)_+ \dd y .
  \label{eq:master}
\end{equation}
The function $y \mapsto g(E^{\!\top}y)$ is a piecewise-linear norm on
$\RR^3$; its unit ball $\{g \le 1\}$ is exactly the image of the unit
cube in contrast space, a zonotope \citep{ziegler1995}.  Writing $y = r\omega$ with
$\omega \in S^2$ and $\dd y = r^2 \dd r \dd\Omega$, and
differentiating in $Y$ with $\dd R/\dd Y = 3/(2R)$,
\begin{equation}
  f(Y) \;=\; F'(Y) \;=\;
  3R \int_{S^2} \bigl(1 - R\,g(\omega)\bigr)\,
    \one\{g(\omega) < 1/R\} \dd\Omega ,
  \label{eq:density}
\end{equation}
a weighted solid-angle integral over the sphere with the
high-range directions excised.  On $S^2$ the range gauge satisfies
$1 \le g(\omega) \le \sqrt2$, with the maximum $\sqrt2$ attained at
the twelve directions $b_{ij} = E(e_i - e_j)$, $i \ne j$---the
$A_3$ root system, of common norm $\sqrt2$---and the minimum $1$ in the six axial-vertex directions
$\pm(e_i + e_j - e_k - e_l)/2$.

Two convex bodies organize everything that follows.

\begin{itemize}[leftmargin=1.6em, itemsep=2pt]
\item \textbf{The cube image (rhombic dodecahedron).}  The support
  $\{g \le 1\}$ of the integrand of~\eqref{eq:master} is the shadow
  of the unit cube under the diagonal projection: a rhombic
  dodecahedron with $12$ rhombic faces at distance $1/\sqrt2$,
  $24$ edges at distance $\sqrt{2/3}$, $8$ ``corner'' vertices at
  distance $\sqrt3/2$ (cube vertices with one or three coordinates
  equal to $1$), and $6$ ``axial'' vertices at distance $1$ (cube
  vertices with exactly two coordinates equal to $1$).  The two cube
  corners on the diagonal itself project to the center.
\item \textbf{The gauge body (cuboctahedron).}  The range gauge is
  the support function of the convex hull of the twelve contrast
  roots $b_{ij}$, which is the \emph{cuboctahedron}; the rhombic
  dodecahedron is its polar dual.  The one nontrivial constant of the
  $n=4$ problem is the angular average
  \begin{equation}
    K_g \;=\; \int_{S^2} g(\omega)\dd\Omega
    \;=\; 12\sqrt{2}\,\arctan\sqrt{2},
    \label{eq:Kg}
  \end{equation}
  i.e.\ $2\pi$ times the cuboctahedron's mean width.  The closed form
  follows from the classical edge formula for the integrated mean
  curvature of a convex polytope \citep{schneider2014},
  $K_g = \frac12\sum_{e}\ell_e\,(\pi - \varphi_e)$: the cuboctahedron
  has $24$ edges of length $\sqrt2$ and a single exterior dihedral
  angle $\arccos(1/\sqrt3) = \arctan\sqrt2$.  The $\sqrt2$'s and
  $\sqrt3$'s pervading the closed forms below are this polytope's
  geometry.
\end{itemize}

\section{Four thresholds, one picture}\label{sec:geometry}

The integrand of~\eqref{eq:density} changes character exactly when
the level set $\{g = 1/R\}$---equivalently, the variance ball
of radius $R = \sqrt{3Y}$---meets a new boundary orbit of the rhombic
dodecahedron.  Since the polytope has four boundary orbits, there are
four critical radii, hence four regimes:

\begin{center}
\renewcommand{\arraystretch}{1.25}
\begin{tabular}{llll}
\toprule
threshold & ball radius $\sqrt{3Y}$ & rhombic-dodecahedron orbit &
  effect on the CDF \\
\midrule
$Y = 1/6$ & $1/\sqrt{2} \approx 0.707$ & $12$ faces (inradius) &
  caps switch on \\
$Y = 2/9$ & $\sqrt{2/3} \approx 0.816$ & $24$ edges &
  wall cuts switch on \\
$Y = 1/4$ & $\sqrt{3}/2 \approx 0.866$ & $8$ corner vertices &
  wall-cut arcs merge \\
$Y = 1/3$ & $1$ & $6$ axial vertices &
  $F = 1$ \\
\bottomrule
\end{tabular}
\end{center}

Figure~\ref{fig:rd-view} shows the four critical distances inside
the polytope; Figure~\ref{fig:rd-sequence} shows the ball at each of
them.  All four assignments were verified numerically from the
projected polytope itself: the distances from the center to the hull's
faces, edges, and two vertex orbits are $0.70710678$, $0.81649658$,
$0.86602540$, and $1$, matching $\sqrt{3Y}$ at the four thresholds to
ten digits.  The first of these distances is classical: that the
nearest boundary of the cube image lies at $1/\sqrt2$ is, for every
$n$, the Lemma of \citet{aliRichards1975}, proved
there algebraically.  The remaining three orbits---edges, corner
vertices, axial vertices---and the regime structure they open are,
to our knowledge, new.

\begin{figure}[t]
\centering
\includegraphics[width=\linewidth]{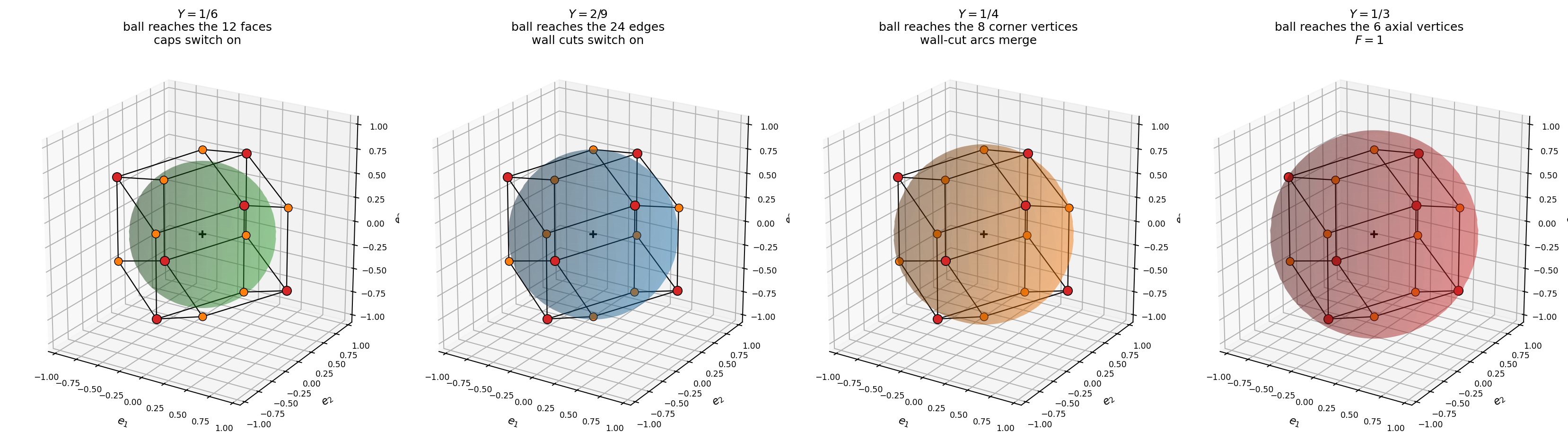}
\caption{The same polytope with the variance ball drawn at each
threshold separately.  Left to right: the ball reaches the $12$ faces
($Y = 1/6$, caps switch on), the $24$ edges ($Y = 2/9$, wall cuts
switch on), the $8$ corner vertices ($Y = 1/4$, the wall-cut arcs
merge), and the $6$ axial vertices ($Y = 1/3$, where $F = 1$).}
\label{fig:rd-sequence}
\end{figure}

\begin{remark}[The ball exits through the roots]\label{rem:roots}
The first exit is dimension-independent in character.  The twelve
face centers of the rhombic dodecahedron lie in the directions of
the normalized contrast roots $u_{ij}$, so the ball first leaves the
cube image through the $n(n-1) = 12$ root directions---precisely the
cap axes of Section~\ref{sec:derive-mid}.  At $n = 3$ the same
statement reads: the disk first leaves the hexagon through its six
edge midpoints, the $n(n-1) = 6$ roots of $A_2$; at $n = 5$ the ball
will exit through the twenty facets of the four-dimensional cube
image, the roots of $A_4$.  In the upper-middle regime
$\tfrac29 < Y < \tfrac14$ the portion of the sphere still inside
the polytope breaks into fourteen islands with full octahedral
symmetry---eight around the corner vertices and six around the
axial vertices.\footnote{Readers of a sporting inclination will
recognize the fourteen-panel layout of the adidas Teamgeist, the
2006 FIFA World Cup match ball, whose design is based on the
truncated octahedron---the order-$4$ permutohedron whose chamber
combinatorics is exactly the $S_4$ structure of
Appendix~\ref{app:trig}.}
\end{remark}

\section{The exact CDF}\label{sec:closedform}

The result has four analytic pieces; the map below links each regime
to its formula and to the geometric event that opens it.

\begin{center}
\renewcommand{\arraystretch}{1.25}
\begin{tabular}{lllll}
\toprule
regime & interval & $F$ & density & geometric event at left endpoint \\
\midrule
low & $\bigl[0,\tfrac16\bigr]$ & \eqref{eq:Flow} & \eqref{eq:flow} &
  --- \\
middle & $\bigl[\tfrac16,\tfrac29\bigr]$ & \eqref{eq:Fmid} &
  \eqref{eq:fmid} & ball exits the $12$ faces; caps switch on \\
upper-middle & $\bigl[\tfrac29,\tfrac14\bigr]$ & \eqref{eq:FR4a} &
  \eqref{eq:fR4a} & ball crosses the $24$ edges; wall cuts \\
high & $\bigl[\tfrac14,\tfrac13\bigr]$ & \eqref{eq:FR4b} &
  \eqref{eq:fR4b} & ball passes the $8$ corner vertices; cuts merge \\
\bottomrule
\end{tabular}
\end{center}

At the right end of the high regime the ball reaches the $6$ axial
vertices and $F(1/3) = 1$.

\subsection{Statement}

Define the four $Y$-dependent angles
\begin{equation}
 A_1=\arctan\tfrac{1}{\sqrt{9Y/2-1}},\quad
 A_2=\arctan\tfrac{1}{\sqrt{18Y-4}},\quad
 A_3=\arccos\tfrac{1}{\sqrt{18Y-3}},\quad
 A_4=\arctan\sqrt{3-\tfrac{2}{3Y}},
 \label{eq:angles}
\end{equation}
and the two auxiliary arctangents
\begin{equation}
  B_\pm(Y) \;=\;
  \arctan\!\frac{3\sqrt{3}\,\sqrt{Y} \pm 2\sqrt{2}}{\sqrt{9Y-2}} ,
  \qquad
  W(Y) \;=\; \arctan\frac{\sqrt2}{\sqrt{9Y-2}} ,
  \label{eq:Bpm}
\end{equation}
where at the endpoint $Y = 2/9$ the quantities $A_1$, $A_2$, $B_+$,
and $W$ are understood as their limiting value $\pi/2$, and $B_-$ as
its limiting value $-\pi/2$, since $\sqrt6 - 2\sqrt2 < 0$.

\begin{theorem}[Exact elementary CDF at $n=4$]\label{thm:main}
For $X_1,\dots,X_4$ i.i.d.\ uniform on $[0,1]$ the CDF
$F(Y)=\Pr(S^2\le Y)$ is, with $K_g = 12\sqrt2\arctan\sqrt2$:

\smallskip
\textbf{Low regime} $0 \le Y \le \tfrac16$ (ball inside the
polytope):
\begin{equation}
  F(Y) \;=\; 8\sqrt{3}\,\pi\,Y^{3/2} \;-\; \tfrac92 K_g\,Y^{2}.
  \label{eq:Flow}
\end{equation}

\textbf{Middle regime} $\tfrac16 \le Y \le \tfrac29$ (caps):
\begin{equation}
  F(Y) \;=\; -\frac{\pi}{\sqrt2}
  \;+\; 18\sqrt{2}\,\pi\,Y
  \;-\; 40\sqrt{3}\,\pi\,Y^{3/2}
  \;+\; \Bigl(54\sqrt{2}\,\pi - \tfrac92 K_g\Bigr) Y^{2}.
  \label{eq:Fmid}
\end{equation}

\textbf{Upper-middle regime} $\tfrac29 \le Y \le \tfrac14$ (caps with
disjoint wall cuts):
\begin{equation}
  F(Y) \;=\; \frac{\pi}{\sqrt2} \;+\; \Phi_a(Y),
  \label{eq:FR4a}
\end{equation}
\begin{align}
  \Phi_a(Y) \;=\;\;& \tfrac43\sqrt{9Y-2}
   \;+\; 18\sqrt2\,Y\,\bigl(\pi - 4A_3\bigr)
   \;+\; 8\sqrt3\,Y^{3/2}\bigl(24 A_4 - 5\pi\bigr)
   \;+\; \sqrt2\,\bigl(B_- - B_+\bigr) \notag\\
  &+\; 27\sqrt2\,Y^{2}\Bigl(\pi
      - 4W
      + 4A_2 + 2\arcsin\tfrac{1}{\sqrt3} - 4A_3\Bigr).
  \label{eq:Phia}
\end{align}

\textbf{High regime} $\tfrac14 \le Y \le \tfrac13$ (merged wall
cuts):
\begin{equation}
  F(Y) \;=\; 1 \;-\; \Phi_b\bigl(\tfrac13\bigr) \;+\; \Phi_b(Y),
  \qquad
  \Phi_b\bigl(\tfrac13\bigr)
   = \tfrac23 + 9\sqrt2\,\arctan\sqrt2 + 4\sqrt2\,\arctan(2\sqrt2),
  \label{eq:FR4b}
\end{equation}
\begin{align}
  \Phi_b(Y) \;=\;\;& \tfrac23\sqrt{9Y-2}
  \;+\; 9\sqrt2\,Y\Bigl(\pi + 2\arcsin\tfrac13 - 4A_3\Bigr)
  \;-\; 24\sqrt3\,Y^{3/2}\bigl(\pi - 4A_4\bigr) \notag\\
  &-\; 9\sqrt2\,Y^{2}\Bigl(\pi + 6\arctan\tfrac{\sqrt2}{5}
     - 4\arcsin\tfrac13 - 6\arcsin\tfrac{1}{\sqrt3}
     - 6\arcsin\tfrac{1}{\sqrt{18Y-3}} \notag\\
  &\hphantom{-\; 9\sqrt2\,Y^{2}\Bigl(\pi\;}
     + 2\arctan\sqrt2
     - 6\arctan\sqrt{\tfrac{9Y}2 - 1}
     + 6\arctan\sqrt{18Y-4}\Bigr) \notag\\
  &+\; \frac{1}{\sqrt2}\Bigl(15A_3 - 8B_- + 8B_+
     + 3\arctan\sqrt{18Y-4}\Bigr).
  \label{eq:Phib}
\end{align}
The four expressions agree at the three interior thresholds; the gap
of \eqref{eq:FR4a} and \eqref{eq:FR4b} at $Y = 1/4$ is symbolically
zero, and $F(1/3) = 1$ holds exactly by construction of the ceiling
anchor in \eqref{eq:FR4b}, and numerically under independent assembly
to $5.7\eee{14}$.
\end{theorem}

Differentiating recovers the density in every regime; the middle and
upper pieces are recorded here because they are the working objects of
the derivation.

\begin{corollary}[Density]\label{cor:density}
With $A_1,\dots,A_4$ as in \eqref{eq:angles},
\begin{align}
 0<Y<\tfrac16:\quad
   f(Y)&=12\sqrt3\,\pi\,\sqrt Y - 9K_g\,Y,
 \label{eq:flow}\\[2pt]
 \tfrac16<Y<\tfrac29:\quad
   f(Y)&= 18\sqrt2\,\pi - 60\sqrt3\,\pi\,\sqrt Y
      + 108\sqrt2\,\pi\,Y - 9K_g\,Y,
 \label{eq:fmid}\\[2pt]
 \tfrac29<Y<\tfrac14:\quad
   f(Y)&=6\pi\bigl(3\sqrt2-10\sqrt3\,\sqrt Y+18\sqrt2\,Y\bigr)
   +36\Bigl[8\sqrt3\,\sqrt Y\,A_4 \notag\\
 &\quad -\sqrt2\bigl(3Y\arctan\sqrt2+(2+6Y)A_3\bigr)
   +6\sqrt2\,Y\,(A_2-A_1)\Bigr],
 \label{eq:fR4a}\\[2pt]
 \tfrac14<Y<\tfrac13:\quad
   f(Y)&=18\Bigl[\pi\bigl(\sqrt2-2\sqrt3\,\sqrt Y+6\sqrt2\,Y\bigr)
   +8\sqrt3\,\sqrt Y\,A_4-2\sqrt2\,(1+3Y)A_3 \notag\\
 &\quad +6\sqrt2\,Y\,(A_2-A_1)-\sqrt2\,(\kappa_0+\kappa_1 Y)\Bigr],
 \label{eq:fR4b}
\end{align}
with $\kappa_0=\arccos\tfrac13$ and
$\kappa_1=4\arccos\tfrac13-2\operatorname{arccot}\sqrt2
+6\arctan\sqrt2+6\arctan\tfrac{\sqrt2}{5}$.
The density is continuously differentiable on $(0,1/3)$ and vanishes
at the ceiling, $f(1/3^-) = 0$.  The thresholds show up one derivative
higher: $f''$ jumps at $1/6$, since
$f_{\mathrm{mid}} - f_{\mathrm{low}} = 18\sqrt2\pi(1-\sqrt{6Y})^2$ has a
double zero there; $f''$ is continuous at $2/9$ but $f'''$ diverges
like $(Y - 2/9)^{-1/2}$ from the right; and $f'''$ jumps at $1/4$.
\end{corollary}

The pattern has a mechanism.  Each threshold is the first contact
of the sphere with a boundary orbit of the zonotope, and the exponent
is set by how the excised measure switches on.  At $1/6$ a cap
appears around each face normal; its area is analytic in
$\cos\psi_R$, so the correction to $F$ begins with an integer power
and shows up as a jump in $f''$.  At $2/9$ the cap circle becomes
tangent to a wall circle, and the wall cut it opens is an arc whose
half-width is an inverse cosine of the colatitude, a square-root
onset, which after the two remaining integrations leaves the
half-power divergence $(Y - 2/9)^{-1/2}$ in $f'''$.  At $1/4$ the
wall cuts of neighboring faces merge; the arc measure acquires a
kink, and $f'''$ jumps.  The single $n = 3$ bifurcation at $Y = 1/4$
is of the second kind: the variance circle touches the hexagon's
edges, the excised arcs open with half-width $\arccos(r/R)$, and the
same half power appears one derivative earlier, in $f''$ rather than
$f'''$, because at $n = 3$ the arc measure is integrated once less.  Figure~\ref{fig:cdf} plots $F$ and $f$; the
density is smooth to the eye, and the thresholds are visible only in
its curvature.

\begin{figure}[t]
\centering
\includegraphics[width=\linewidth]{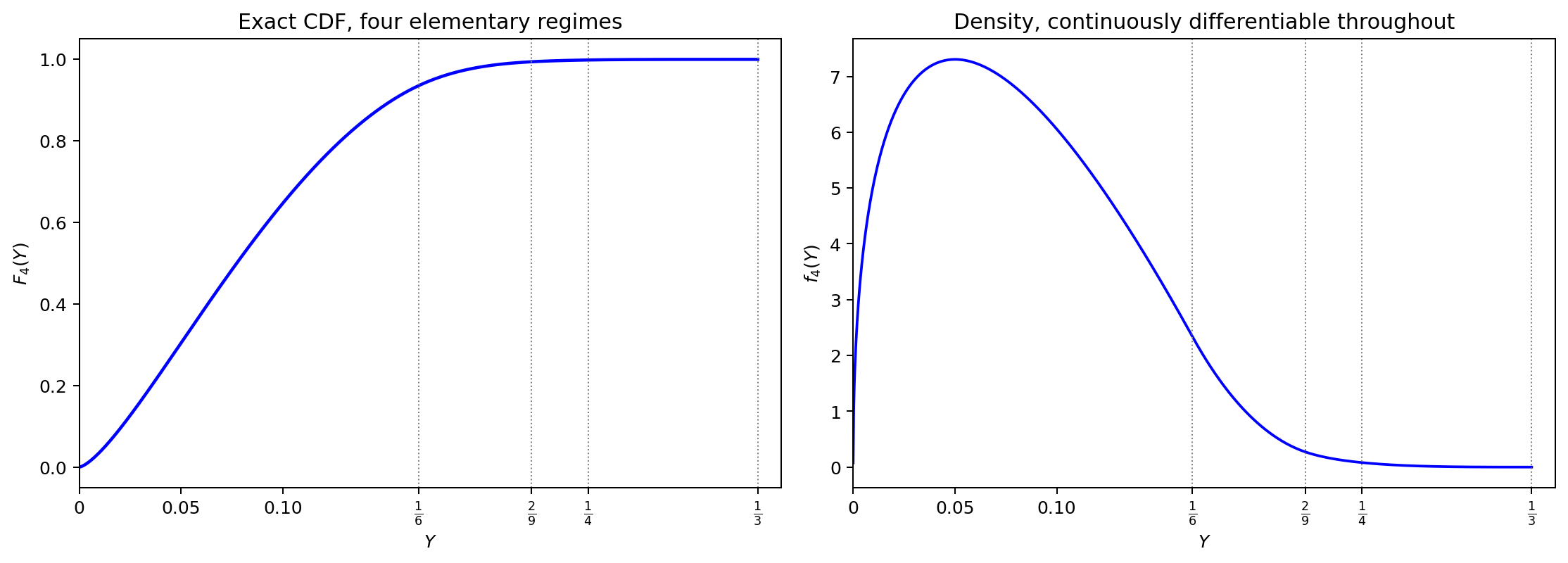}
\caption{The exact $n=4$ CDF (left) and its density (right), from the
closed forms of Theorem~\ref{thm:main}.  Dotted verticals mark the
thresholds $1/6$, $2/9$, $1/4$ and the ceiling $1/3$.}
\label{fig:cdf}
\end{figure}

\subsection{Proof of Theorem~\ref{thm:main}: low regime}\label{sec:derive-low}

The proof of Theorem~\ref{thm:main} occupies this and the next two subsections, one regime at
a time, each starting from the master integral~\eqref{eq:master}.
For $R\sqrt2 \le 1$, i.e.\ $Y \le 1/6$, the factor
$\bigl(1 - g\bigr)_+$ never clips inside the ball, because
$r\,g(\omega) \le R\sqrt2 \le 1$.  Then \eqref{eq:master} splits into
a volume term and a gauge moment:
\begin{equation}
  F(Y) = 2\Bigl[\operatorname{vol}(B_R)
     - \int_0^R r^3 \dd r \int_{S^2} g\dd\Omega\Bigr]
   = 2\Bigl[\tfrac{4\pi}{3}R^3 - \tfrac{R^4}{4}K_g\Bigr]
   = 8\sqrt3\,\pi\,Y^{3/2} - \tfrac92 K_g\,Y^2 ,
\end{equation}
using $R^3 = 3\sqrt3\,Y^{3/2}$ and $R^4 = 9Y^2$.  This is
\eqref{eq:Flow}; the entire low regime is a two-term formula, and the
small-$Y$ law $F \sim 8\sqrt3\pi Y^{3/2}$ has the universal exponent
$(n-1)/2$; it is the $n = 4$ case of Royen's thin-cylinder
asymptotic for the uniform parent \citep{royen2008}.

\begin{remark}[The Ali--Richards constant]\label{rem:AR}
Equation \eqref{eq:Flow} recovers the low-regime formula of \citet{aliRichards1975}.  Their exact values,
$b(4) = 6\sqrt2\arccos(1/\sqrt3)$ and
$b(5) = 8\sqrt{10}\bigl(\arccos(1/\sqrt3) - \pi/6\bigr)$, coincide
with ours identically---$\arccos(1/\sqrt3) = \arctan\sqrt2$---and
our derivation identifies the constant geometrically:
\begin{equation*}
  b(n) = \frac{1}{\sqrt n}\int_{S^{n-2}} g\dd\Omega,
  \qquad
  b(4) = \tfrac12 K_g = 6\sqrt2\arctan\sqrt2,
  \qquad
  b(5) = K_5/\sqrt5,
\end{equation*}
where $K_5 = 20\sqrt2\bigl(\tfrac{2\pi}{3} - \arccos\tfrac13\bigr)
= \sqrt5\,b(5)$ is Ali and Richards' constant rewritten in the
normalization of $K_g$.  The general $b(n)$ is therefore a
dimensional constant times the mean width of the convex hull of the
contrast roots $b_{ij}$.\footnote{The six-decimal table in
\citet{aliRichards1975} prints $b(4) = 8.106126$ and
$b(5) = 10.921648$; their own closed forms evaluate to
$8.106130\ldots$ and $10.921694\ldots$, so the printed decimals
carry small transcription errors while the formulas are exact.}
\end{remark}

\subsection{Proof of Theorem~\ref{thm:main}: middle regime}\label{sec:derive-mid}

For $1/6 < Y < 2/9$ the excision set $\{g \ge 1/R\}$ in
\eqref{eq:density} is the disjoint union of twelve congruent geodesic
caps centered at the roots $b_{ij}$, on which $g = \sqrt2\cos\psi$ in
the colatitude $\psi$ from the cap axis.  Permutations of the four
coordinates act on contrast space as the tetrahedral reflection group
$T_d \cong S_4$, a subgroup of the polytope's full symmetry group
$O_h$, carrying the caps onto one another isometrically, so one cap
determines all twelve.  On the sphere $r = R$ the clipping factor $\bigl(1 - r\,g(\omega)\bigr)_+$ of \eqref{eq:master} is active
where $\sqrt2\,R\cos\psi \ge 1$, so each cap has angular radius
\begin{equation}
  \psi_R = \arccos\frac{1}{\sqrt2\,R}, \qquad R = \sqrt{3Y},
  \label{eq:psiR}
\end{equation}
growing from $\psi_R = 0$ at $Y = 1/6$, where the ball first touches the faces, to
$\psi_R = \pi/6$ at $Y = 2/9$, where the caps reach the walls.  Then
\begin{equation}
  f(Y) = 3R\Bigl[\,4\pi - R\,K_g
   - 12\!\int_0^{\psi_R}\!\!
     \bigl(1-\sqrt2\,R\cos\psi\bigr)\,2\pi\sin\psi\dd\psi\Bigr],
  \label{eq:midcap}
\end{equation}
and the cap integral is elementary:
$2\pi\bigl[(1-c) - \tfrac{\sqrt2 R}{2}(1-c^2)\bigr]$ with
$c = 1/(\sqrt2 R)$.  Expanding \eqref{eq:midcap} gives exactly
\eqref{eq:fmid}; we verified the identity numerically to
$2\eee{14}$ across the regime.  Integrating and matching to
\eqref{eq:Flow} at $Y = 1/6$ fixes the constant of integration to
$-\pi/\sqrt2$ exactly, giving \eqref{eq:Fmid}.

\subsection{Proof of Theorem~\ref{thm:main}: upper regimes}\label{sec:derive-high}

Past $Y = 2/9$ the ball reaches the polytope's edges: each cap grows
until it meets the four \emph{walls} of its isotropy cell---the great
circles $\{c_i = c_k\}$ on which the identity of the extreme
coordinate pair changes.  The wall circles excise from each cap four
azimuthal arcs of half-width
\begin{equation}
  \beta(\psi) \;=\; \arccos\!\Bigl(\frac{\cot\psi}{\sqrt3}\Bigr),
  \label{eq:beta}
\end{equation}
and every resulting integral carries only the single quadratic
radical $\sqrt{3-4\cos^2\psi}$---the algebraic signature of the
$60^\circ$ wall geometry of the $A_3$ arrangement.  The arcs on
neighboring walls stay disjoint until $\psi = \operatorname{arccot}
\sqrt2$, i.e.\ until the ball reaches the corner vertices at
$Y = 1/4$, which splits the computation into the two pieces
\eqref{eq:fR4a} and \eqref{eq:fR4b}.  Carrying out the arc integrals
produces the densities of Corollary~\ref{cor:density}; they were
derived by computer algebra and match the directly assembled
spherical integral to $8\eee{14}$ at every spot-checked point.
Figure~\ref{fig:atlas-sphere} shows the mechanism in the simplest picture: seen from the
center of the projected cube, a cell is one rhombic face, its four walls are the face's edges,
and the ball pokes through the face, crosses its edges at $Y = 2/9$, passes its two near
corners at $Y = 1/4$, and reaches its two far corners at $Y = 1/3$.

\begin{figure}[t]
\centering
\includegraphics[width=\linewidth]{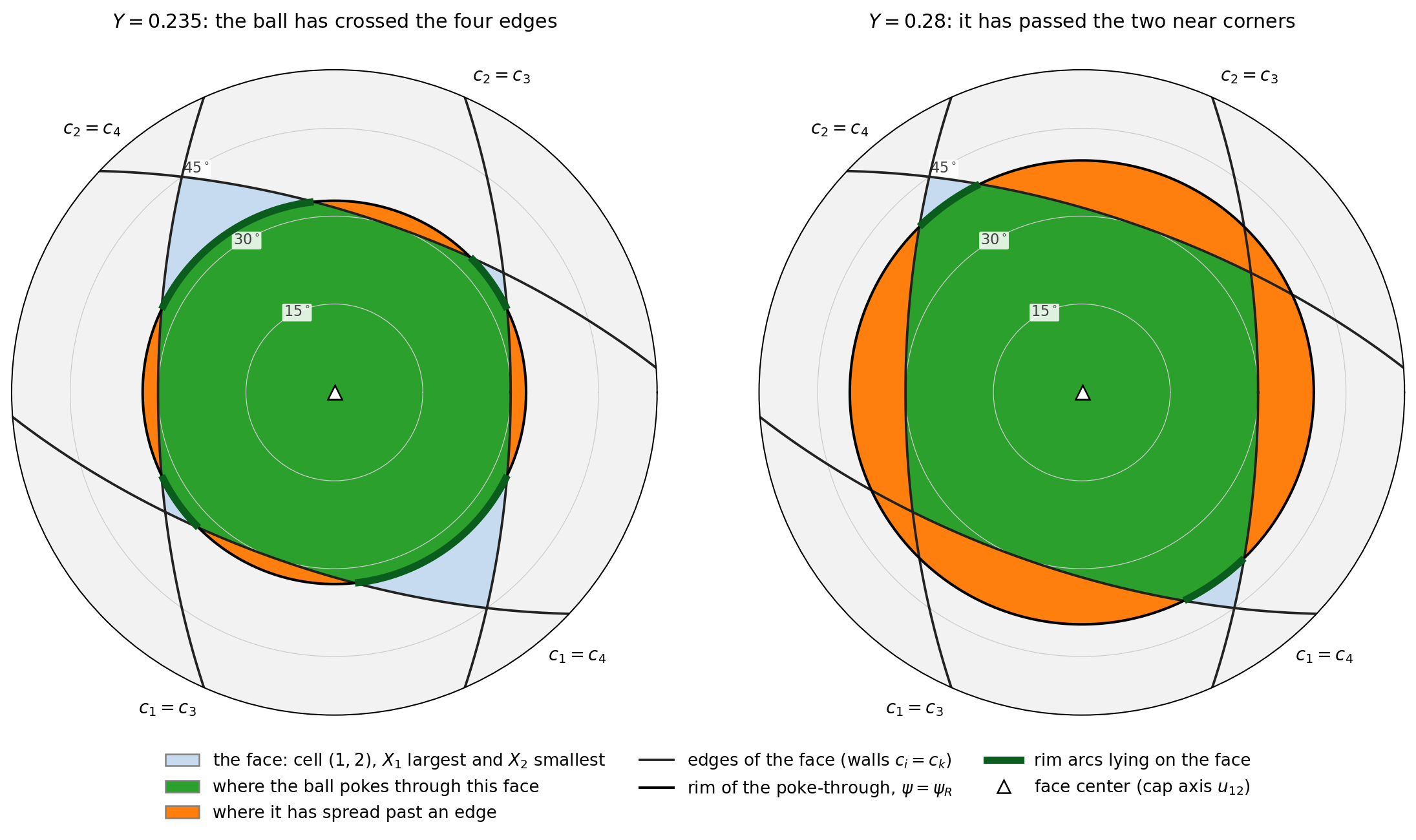}
\caption{What happens on one face as $Y$ grows.  Stand at the center of the projected cube
(the rhombic dodecahedron) and look straight at one of its twelve rhombic faces, the one whose
directions form the cell $(1,2)$: in those directions $X_1$ is the largest coordinate and $X_2$
the smallest.  Each panel is a map of directions around the face center (triangle): distance
from the center of the map is the angle $\psi$ from the face's normal $u_{12}$, and angle
around it is the azimuth $\varphi$.  The face is the blue region; its four edges are the walls
$c_i = c_k$, which look curved only because the map is drawn in angles.  The ball
$\{S^2 \le Y\}$ pokes through the face in the green region, whose rim is the circle
$\psi = \psi_R$ of \eqref{eq:psiR}; where the poke-through has spread past an edge (orange),
it is counted by the neighboring face instead.  What the cap integral needs is the part of the
rim that still lies on the face (thick dark green).  \emph{Left,} $Y = 0.235$: the ball has
crossed all four edges, which it first reaches at $Y = 2/9$, $\psi_R = 30^\circ$, but not yet
the corners, so four separate arcs of rim survive.  \emph{Right,} $Y = 0.28$: past
$Y = 1/4$ ($\psi_R = 35.26^\circ$) the ball has passed the two near corners, where pairs of
edges meet, the neighboring orange regions have merged there, and only two arcs survive.  At
$Y = 1/3$ ($\psi_R = 45^\circ$) it reaches the two far corners and the face is covered.}
\label{fig:atlas-sphere}
\end{figure}

The antiderivatives $\Phi_a, \Phi_b$ of \eqref{eq:Phia} and
\eqref{eq:Phib} were obtained by computer algebra and are certified
independently of it: differentiating them in $60$-digit arithmetic
reproduces \eqref{eq:fR4a} and \eqref{eq:fR4b} to $10^{-58}$ at
interior points; the matching constant at $Y = 2/9$ simplifies
symbolically to exactly $\pi/\sqrt2$; and the gap between
\eqref{eq:FR4a} and \eqref{eq:FR4b} at $Y = 1/4$ is symbolically zero
and below $10^{-60}$ in $60$-digit evaluation.  The high-regime
branch is anchored at the ceiling: $F(1/3) = 1$ holds by
construction, with $\Phi_b(1/3)$ in the closed form
\eqref{eq:FR4b}.  Independently, numerical integration of the density
chain from $0$ to $1/3$ returns $1$ to $5\eee{14}$, so the anchor is
consistent with the differential structure rather than imposed on it.

This completes the proof of Theorem~\ref{thm:main}.

\begin{remark}[Structure of the constants]\label{rem:constants}
The matching constants at the three joints are $-\pi/\sqrt2$,
$+\pi/\sqrt2$, and the ceiling combination
$\tfrac13 - 9\sqrt2\arctan\sqrt2 - 4\sqrt2\arctan(2\sqrt2)$.  Every
inverse-trigonometric constant in the $n = 4$ problem---%
$\arctan\sqrt2$, $\arctan(2\sqrt2)$, $\arctan(\sqrt2/5)$,
$\arccos\tfrac13$, $\arcsin\tfrac13$---is an angle of the
cuboctahedron--rhombic-dodecahedron pair, and each is a rational
combination of $\pi$ and $\arctan\sqrt2$:
$\arccos\tfrac13 = \arctan(2\sqrt2) = \pi - 2\arctan\sqrt2$,
$\arctan(\sqrt2/5) = \pi - 3\arctan\sqrt2$, and
$\arcsin\tfrac13 = 2\arctan\sqrt2 - \pi/2$.  This is the sense in
which $K_g$ is the single non-obvious constant.  In particular $F(1/4)$
evaluates to $0.9985125531\ldots$: only $1.5\eee{3}$ of probability
mass lives beyond the corner-vertex threshold.
\end{remark}

\subsection{The ceiling law}\label{sec:ceiling}

At the top of the support the tail has a simple exact leading term, which the closed form
reproduces and which extends to every polynomial parent.

\begin{corollary}[Ceiling law]\label{cor:ceiling}
Write $\delta = \tfrac13 - Y$ for the distance of the variance level below the ceiling.  As
$\delta \to 0^+$,
\begin{equation}
  1 - F\bigl(\tfrac13 - \delta\bigr) \;=\; \tfrac{81}{4}\,\delta^4
  \;+\; O(\delta^5).
  \label{eq:ceiling-law}
\end{equation}
\end{corollary}

\begin{proof}
The maximum $Y = 1/3$ is attained only at the six axial vertices such
as $(1,1,0,0)$.  Writing a nearby sample as $(1-\epsilon_1, 1-\epsilon_2, \epsilon_3, \epsilon_4)$ with
$\epsilon_i \ge 0$ small, $\sum_i (x_i - \bar x)^2 = 1 - (\epsilon_1+\epsilon_2+\epsilon_3+\epsilon_4)
+ O(\|\epsilon\|^2)$ exactly to first order, so the event
$\{Y > 1/3 - \delta\}$ near that vertex is the simplex
$\{\epsilon_1+\epsilon_2+\epsilon_3+\epsilon_4 < 3\delta\}$ up to the quadratic remainder, of
volume $(3\delta)^4/4!$; six vertices give $81\delta^4/4$.  The
gradient of the variance at the vertex is nonzero, so the level set is
a simplex and not an ellipsoidal shell, which is why the exponent is
$n = 4$ and not a half-integer.
\end{proof}

\citet{royen2008} obtains the leading term for the uniform parent
at every $n$ by the same corner-simplex argument, and at $n = 4$ his
formula gives exactly $81/4$; we include the proof because the same
argument yields the polynomial-parent tail below.  The closed form
\eqref{eq:FR4b} reproduces the coefficient: the ratio
$(1 - F(1/3 - \delta))/\delta^4$ evaluates to $20.5$ at
$\delta = 3\times10^{-3}$ and $20.3$ at $10^{-3}$ against
$81/4 = 20.25$, the residual being the $O(\delta)$ correction.  The same
simplex argument applies to every polynomial parent.  If $f(x) \sim c_0 x^{k_0}$ as $x \to 0$
and $f(1-\epsilon) \sim c_1 \epsilon^{k_1}$ as $\epsilon \to 0$, with $k_0, k_1 \ge 0$ the orders
of vanishing at the endpoints, the joint density near an axial vertex behaves as
$c_0^2 c_1^2 (\epsilon_1\epsilon_2)^{k_1}(\epsilon_3\epsilon_4)^{k_0}$, and the Dirichlet integral over
the simplex gives
\begin{equation}
  1 - F\bigl(\tfrac13 - \delta\bigr) \;=\;
  \frac{6\,c_0^2 c_1^2\,(k_0!)^2 (k_1!)^2}{m!}\,(3\delta)^{m} + O\bigl(\delta^{m+1}\bigr),
  \qquad m = 4 + 2k_0 + 2k_1 .
  \label{eq:ceiling-poly}
\end{equation}
The uniform parent ($k_0 = k_1 = 0$, $c_0 = c_1 = 1$) recovers $81/4$; Beta$(2,1)$ ($f = 2x$)
and Beta$(2,2)$ ($f = 6x(1-x)$) give the leading tails $\tfrac{486}{5}\delta^6$ and
$\tfrac{177147}{140}\delta^8$ quoted below; and a Beta$(a,b)$ parent, polynomial or not, has
the exponent $2a + 2b$ by the same computation with $k_0 = a - 1$, $k_1 = b - 1$.  Near the ceiling the
tail should be computed from this expansion rather than as $1 - F$: in
double precision the closed forms for the Beta$(2,1)$ and
Beta$(2,2)$ parents evaluate to $1 + O(10^{-13})$ at
$Y = 1/3 - 10^{-3}$, where the true tails, $\tfrac{486}{5}\delta^6$
and $\tfrac{177147}{140}\delta^8$ to leading order, are of order
$10^{-17}$ and $10^{-21}$.

\section{Validation}\label{sec:validation}

The closed form was checked against every independent computation of
$F$ available to us.

\paragraph{Deterministic quadrature.}
An independent tensor Gauss--Legendre quadrature of
$F(Y)=\operatorname{vol}\{X\in[0,1]^4:\sum_i(X_i-\bar X)^2\le 3Y\}$
in the original coordinates, with $480$ nodes per coordinate and the
$X_4$-integration done exactly, agrees with Theorem~\ref{thm:main} to
$\le 1.3\eee{7}$ at every tabulated $Y$, within its own
grid-refinement change of at most $2.7\eee{7}$, and to
$\le 2\eee{10}$ at $Y = 0.235$, $0.26$ and $0.30$ in the two upper
regimes:

\begin{center}
\renewcommand{\arraystretch}{1.15}
\begin{tabular}{lcccccc}
\toprule
$Y$ & $0.04$ & $0.10$ & $1/6^-$ & $0.18$ & $0.20$ & $2/9^-$ \\
\midrule
$F$ closed form  & $0.2315212$ & $0.6470252$ & $0.9353894$ &
  $0.9611738$ & $0.9836197$ & $0.9941677$ \\
$F$ quadrature & $0.2315213$ & $0.6470251$ & $0.9353894$ &
  $0.9611738$ & $0.9836197$ & $0.9941677$ \\
$|\mathrm{diff}|$ & $9.6\eee8$ & $1.3\eee7$ & $9.3\eee9$ &
  $1.7\eee9$ & $8.5\eee{10}$ & $1.1\eee9$ \\
\bottomrule
\end{tabular}
\end{center}

\paragraph{Monte Carlo.}
At $N = 2\times10^7$ uniform draws the assembled explicit CDF agrees
with the empirical CDF in the upper regimes to
$|\Delta| \le 4.1\eee{6}$, within binomial noise at these
probabilities.

\paragraph{Internal exactness checks.}
The density chain \eqref{eq:flow}--\eqref{eq:fR4b} is continuous at
the three thresholds under numerical evaluation with gaps
$\le 3\eee{8}$ set by evaluation offsets, vanishes at the ceiling
with $f(1/3^-) = -4.8\eee{14}$, and integrates to
$F(1/3) = 1 + 5.7\eee{14}$.  The explicit-form joint at $Y = 1/4$ and
the matching constant $\pi/\sqrt2$ at $Y = 2/9$ are symbolic
identities, not numerical fits.  The first two moments of $S^2$
computed from Theorem~\ref{thm:main},
$\int_0^{1/3}(1-F)\dd Y$ and $\int_0^{1/3}2Y(1-F)\dd Y$, equal the
exact values $1/12$ and $41/4320$ to $30$ digits.

\section{Polynomial parents inherit the closure}\label{sec:poly}

The companion $n = 3$ paper \citep{n3paper} proves a closure theorem:
every polynomial parent density on $[0,1]$ yields an elementary
three-sample variance CDF.  The $n = 4$ analog holds in the same
generality.

\begin{theorem}[Polynomial closure at $n = 4$]\label{thm:poly}
Let $f$ be any polynomial probability density on $[0,1]$ and let
$F(Y) = \Pr(S^2 \le Y)$ for four i.i.d.\ observations from $f$.
Then $F$ is elementary on each of the four regimes
$[0,\tfrac16]$, $[\tfrac16,\tfrac29]$, $[\tfrac29,\tfrac14]$ and
$[\tfrac14,\tfrac13]$ of the uniform case, and the proof gives a
finite construction of it.
\end{theorem}

The unit interval carries no loss of generality.  If $X = A + (B-A)U$ with $U$ supported on
$[0,1]$, then $S_X^2 = (B-A)^2 S_U^2$, so $\Pr(S_X^2 \le y) = F_U\bigl(y/(B-A)^2\bigr)$, and a
polynomial density on $[A,B]$ is carried to a polynomial density on $[0,1]$.  Every statement
below for parents on $[0,1]$ therefore holds on any bounded interval, with the thresholds scaled
by $(B-A)^2$.

The proof is given in Section~\ref{sec:poly-proof}, once the
mean-slab weight has been set up.  It does not integrate in the
radius at all: the divergence theorem converts every polynomial
moment of the ball clipped by the support polytope into boundary
integrals, which are trigonometric polynomials along circular arcs,
areas of spherical regions bounded by such arcs, and polynomial
integrals over planar regions bounded by segments and arcs.

Besides the proof, the rest of this section supplies
\emph{constructive} content: an algorithm carried out in full for two
Beta parents on the entire support.  The regime
thresholds are parent-independent because they are properties of the
support zonotope, not of the weight.

\subsection{Structure}

For a polynomial parent $f$ the mean-slab weight
\begin{equation}
  \rho(y) \;=\; 2\int_{-\min_i w_i}^{\,1-\max_i w_i}
     \prod_{i=1}^4 f(\mu + w_i)\dd\mu,
  \qquad w = E^{\!\top}y,
\end{equation}
is, on each of the $24$ ordering chambers of contrast space, a
polynomial in $w$---the integrand is a polynomial in $\mu$ and the
$w_i$, and the slab endpoints are linear in $\max_i w_i$ and
$\min_i w_i$.  In the low regime $Y \le 1/6$ the ball
$\{\|y\| \le R\}$ stays inside the support zonotope, so no clipping
occurs and the CDF is the \emph{finite} sum
\begin{equation}
  F(Y) \;=\; \sum_{d} \mathcal{A}_d\,\frac{R^{\,d+3}}{d+3},
  \qquad R = \sqrt{3Y},
  \qquad
  \mathcal{A}_d = \int_{S^2} \rho_d(\omega)\dd\Omega
      = 24\int_{T} \rho_d(\omega)\dd\Omega,
  \label{eq:poly-structure}
\end{equation}
where $\rho_d$ is the degree-$d$ homogeneous part of $\rho$ evaluated
on sorted contrasts and $T$ is the ordered chamber, a spherical
triangle of area $\pi/6$.  For $Y > 1/6$ the ball crosses the zonotope and the cap and
wall-cut integrals of Sections~\ref{sec:derive-mid} and~\ref{sec:derive-high} return with
polynomial weights; Theorem~\ref{thm:poly} covers all four regimes (proof in
Section~\ref{sec:poly-proof}), and Section~\ref{sec:poly-upper} gives the full-support forms.
Each $\mathcal{A}_d$ is a polynomial moment over a
spherical triangle, and such moments are computable \emph{exactly} by
the spherical divergence recursion: for unit-sphere moments of
products of linear forms,
\begin{equation}
  (d+1)\!\int_T \prod_{i=1}^{d}(a_i\!\cdot\!\omega)\dd\Omega
  = \sum_{j=2}^{d}(a_1\!\cdot\!a_j)\!\int_T\prod_{i\ge2,\,i\ne j}\!
    (a_i\!\cdot\!\omega)\dd\Omega
  \;-\; \sum_{e}\,(a_1\!\cdot\!n_e)\!\int_{e}\prod_{i\ge2}
    (a_i\!\cdot\!\omega)\dd\ell,
  \label{eq:div-recursion}
\end{equation}
with $n_e$ the constant outward conormal of geodesic edge $e$; the
edge integrals are elementary trigonometric integrals over arcs of
length $\arccos$ of algebraic numbers.  The recursion is a special case
of Arvo's recurrence for polynomial moments over regions of the
sphere \citep{arvo1995}, which also allows a varying boundary conormal,
as in \eqref{eq:div-recursion-general} below.  Analytic integration of
spherical harmonics over spherical polygons \citep{wangRamamoorthi2018,belcour2018} treats fixed polygonal domains; the domains here are
chambers clipped by small circles that move with $R$.  Reducing such angular
moments to Gaussian integrals is likewise anticipated
in \citet{aliRichards1975}, where $\Gamma(n/2)\,b(n)$ is evaluated
as a Gaussian moment over the ordered chamber.  For the two parents
worked out below the moments show a parity pattern: even-degree
$\mathcal{A}_d$ lie in $\mathbb{Q}\,\pi + \mathbb{Q}\sqrt3$, the numbers $a\pi + b\sqrt3$
with $a, b$ rational, and odd-degree $\mathcal{A}_d$ in
$\mathbb{Q} + \mathbb{Q}\,\sqrt2\arctan\sqrt2$, the numbers $a + b\sqrt2\arctan\sqrt2$, the latter carrying
the cuboctahedral constant $K_g$ that governs the uniform case.  The
pattern is a property of these assembled slab-weight moments, not of
chamber moments in general: with $c_1 \ge c_2 \ge c_3 \ge c_4$ the
sorted contrast coordinates of $\omega \in T$,
$24\int_T c_1 c_4 \dd\Omega = -\pi/3 - 4$, an even-degree moment
outside $\mathbb{Q}\,\pi + \mathbb{Q}\sqrt3$.  For a parent whose
coefficients are not rational the coefficients themselves enter the
constants, so any statement about the constant field of $F$ refers to
a stated parent; for the two parents below every constant lies in
$\mathbb{Q}(\pi,\sqrt2,\sqrt3,\arctan\sqrt2)$, the numbers built from these four constants and
the rationals by the arithmetic operations, certified symbolically
in Section~\ref{sec:poly-upper}.

\subsection{Proof of Theorem~\ref{thm:poly}}\label{sec:poly-proof}

Write $Z = \{y : g(E^{\!\top}y) \le 1\}$ for the support zonotope and,
for each of the $24$ ordering chambers $C_\sigma$ of contrast space,
$P_\sigma = C_\sigma \cap Z$, a convex polytope with a vertex at the
origin.  On $P_\sigma$ the weight $\rho$ is a polynomial $p_\sigma$,
as noted above, and $\rho = 0$ outside $Z$, so
\begin{equation}
  F(Y) \;=\; \sum_\sigma \int_{P_\sigma \cap B_R} p_\sigma(y)\dd y,
  \qquad R = \sqrt{3Y}.
  \label{eq:F-polytopes}
\end{equation}
Splitting each $p_\sigma$ into homogeneous parts reduces the theorem
to the following lemma, applied to finitely many pairs.

\begin{lemma}\label{lem:ball-polytope}
Let $P \subset \RR^3$ be a convex polytope with facets $\Phi_j$,
outward unit normals $n_j$ and offsets $h_j$, so that
$P = \bigcap_j \{x : n_j\cdot x \le h_j\}$, and let $p$ be a
homogeneous polynomial of degree $d$.  Put
$B_R = \{\|x\| \le R\}$ and
$\Omega_R = \{\omega \in S^2 : R\omega \in P\}$.  Then
\begin{equation}
  \int_{P \cap B_R} p \dd x
  \;=\; \frac{1}{d+3}\Bigl[\, R^{d+3}\int_{\Omega_R} p(\omega)\dd\Omega
    \;+\; \sum_j h_j \int_{\Phi_j \cap B_R} p \dd A \Bigr],
  \label{eq:ball-polytope}
\end{equation}
and on every open interval of $R$ that contains none of the
distances from the origin to the vertices of $P$ and to the feet of
the perpendiculars from the origin onto its edges and facets, the
right-hand side is a finite sum of products of algebraic functions
of $R$ and inverse trigonometric functions of algebraic functions
of $R$.
\end{lemma}

\begin{proof}
Since $p$ is homogeneous of degree $d$,
$\operatorname{div}(x\,p) = 3p + x\cdot\nabla p = (d+3)\,p$, the homogeneity identity behind
the facet reduction of polynomial integrals over polytopes \citep{lasserre1998}.  Apply
the divergence theorem on the Lipschitz domain $P \cap B_R$; the spherical part of its boundary
is what the ball adds.  On the
spherical part of its boundary, $x = R\omega$ with
$\omega \in \Omega_R$, the outward normal is $\omega$,
$x\cdot\omega = R$, $\dd A = R^2\dd\Omega$ and
$p(x) = R^d p(\omega)$, which gives the first term; on
$\Phi_j \cap B_R$ the outward normal is $n_j$ and $x\cdot n_j = h_j$,
which gives the second.  No integration in $R$ is involved, and
facets through the origin, with $h_j = 0$, do not contribute.

\emph{Combinatorics.}  $\Omega_R$ is the intersection of $S^2$ with
the half-spaces $n_j\cdot\omega \le h_j/R$.  Its boundary consists of
arcs of the circles $n_j\cdot\omega = h_j/R$, great circles when
$h_j = 0$ and small circles otherwise, and its corners are the points
$R^{-1}(e \cap \partial B_R)$ for edges $e$ of $P$.  A point of
$e \cap \partial B_R$ solves a quadratic equation along $e$, so its
coordinates are algebraic functions of $R$.  The arrangement of arcs,
corners and components of $\Omega_R$, and the combinatorial type of
each planar region $\Phi_j \cap B_R$, change only when $\partial B_R$
passes a vertex of $P$ or becomes tangent to the line of an edge or
the plane of a facet, that is, at the excluded radii.

\emph{Spherical term.}  Write $p$ as a combination of products of
linear forms.  For a product $\prod_{i=1}^k (a_i\cdot\omega)$ and a
region $\Omega \subset S^2$ with piecewise smooth boundary and
outward conormal $\nu$, the divergence theorem on the sphere, applied
to the tangential field
$\bigl(a_1 - (a_1\cdot\omega)\,\omega\bigr)\prod_{i\ge2}(a_i\cdot\omega)$,
gives
\begin{equation}
  (k+1)\!\int_\Omega \prod_{i=1}^{k}(a_i\!\cdot\!\omega)\dd\Omega
  = \sum_{j=2}^{k}(a_1\!\cdot\!a_j)\!\int_\Omega
    \prod_{i\ge2,\,i\ne j}(a_i\!\cdot\!\omega)\dd\Omega
  \;-\; \int_{\partial\Omega}(a_1\!\cdot\!\nu)
    \prod_{i\ge2}(a_i\!\cdot\!\omega)\dd\ell,
  \label{eq:div-recursion-general}
\end{equation}
which is \eqref{eq:div-recursion} with a conormal that may vary along
the boundary, a case covered by Arvo's recurrence \citep{arvo1995}.  Each application lowers the degree of the interior
integrand by two, so the recursion ends in boundary integrals and,
for even $k$, in the area of $\Omega$.  On an arc of the circle
$n\cdot\omega = c$ with $|c| < 1$, parametrize
$\omega(t) = c\,n + s\,(u\cos t + v\sin t)$ with $s = \sqrt{1-c^2}$
and $(n,u,v)$ orthonormal.  Then $\dd\ell = s\dd t$ and
$\nu = \pm(n - c\,\omega)/s$, so
$(a_1\cdot\nu)\dd\ell = \pm(a_1\cdot n - c\,a_1\cdot\omega)\dd t$: the
factor $s$ cancels, and the boundary integrand is a trigonometric
polynomial in $t$ with coefficients polynomial in $c$ and $s$.  Its
integral between two corners is a polynomial in their parameters and
in the sines and cosines of those parameters, which are algebraic
functions of $R$, as is $c = h_j/R$.  The area is given by the
Gauss--Bonnet formula,
\[
  \operatorname{area}(\Omega) \;=\; 2\pi\chi(\Omega)
  \;-\; \sum_{\text{corners}}\theta_{\mathrm{ext}}
  \;-\; \int_{\partial\Omega}\kappa_g\dd\ell ,
\]
in which each exterior corner angle is the arccosine of an inner
product of two conormals, an algebraic function of $R$, and the
geodesic curvature is constant along each arc, equal to $\pm c/s$, so
that the last integral is a sum of terms $\pm c\,\Delta t$.  The
Euler characteristic of $\Omega_R$ is constant on the interval.

\emph{Planar terms.}  $\Phi_j \cap B_R$ is the facet polygon
intersected with the disk of radius $\sqrt{R^2 - h_j^2}$ centered at
$h_j n_j$.  In coordinates $(z_1, z_2)$ on the facet plane, $p$
restricts to a polynomial $q$; with a polynomial $Q$ satisfying
$\partial Q/\partial z_1 = q$, Green's theorem gives
$\int q \dd A = \oint Q \dd z_2$.  On a straight boundary segment
this is a polynomial in the endpoint coordinates, which are algebraic
in $R$; on a circular arc $z = z_0 + r(\cos t, \sin t)$ it is the
integral of a trigonometric polynomial, handled as above.
\end{proof}

\begin{proof}[Proof of Theorem~\ref{thm:poly}]
By \eqref{eq:F-polytopes} and the lemma, applied to each homogeneous
part of each $p_\sigma$, $F$ is elementary in $R = \sqrt{3Y}$, hence
in $Y$, on every interval that avoids the radii the lemma excludes
for the polytopes $P_\sigma$.  Every face of $P_\sigma$ lies in a
face of $Z$, in a chamber wall, or in the intersection of a wall with
a face of $Z$, and the facets lying in chamber walls pass through the
origin (the intersection of a wall with a face of $Z$ need not).
Each wall is the mirror of a coordinate transposition, which
preserves $Z$, so a wall crosses the relative interior of a face $G$
of $Z$ only if the reflection maps $G$ to itself; the foot of the
perpendicular from the origin onto the affine hull of $G$ is then
fixed by the reflection, so it lies on the wall and is also the foot
onto the part of the wall inside that affine hull.  Hence every
nonzero radius excluded for $P_\sigma$ is the distance from the
origin to a vertex of $Z$ or to the foot of the perpendicular onto an
edge or facet of $Z$, and these are the four radii of
Section~\ref{sec:geometry}: $1/\sqrt2$, $\sqrt{2/3}$, $\sqrt3/2$ and
$1$.  At those radii $F$ is continuous, being the integral of a
bounded weight over a continuously varying domain, so each regime's
expression extends to its endpoints.  The construction is finite:
finitely many chambers, homogeneous degrees, recursion steps, arcs
and corners.  In the low regime $R \le 1/\sqrt2$ the ball lies
inside $Z$, each $\Omega_R$ is a chamber triangle, the planar terms
vanish, and \eqref{eq:ball-polytope} reduces to
\eqref{eq:poly-structure}.
\end{proof}

\subsection{Worked example: Beta\texorpdfstring{$(2,1)$}{(2,1)}}

For a Beta$(a,b)$ parent, with density $f(x) = x^{a-1}(1-x)^{b-1}/B(a,b)$ on $[0,1]$, we
denote by
\begin{equation*}
  F^{(a,b)}(Y) \;=\; \Pr\bigl(S^2 \le Y\bigr), \qquad X_1,\dots,X_4 \ \text{i.i.d.\ Beta}(a,b),
\end{equation*}
the distribution function of the sample variance of four observations; the uniform case of
Theorem~\ref{thm:main} is $F^{(1,1)}$.  For $f(x) = 2x$, the Beta$(2,1)$ density, the weight $\rho$ has homogeneous degrees
$\{0,1,2,3,5\}$, and the recursion gives
\begin{equation}
  \mathcal{A}_0 = \tfrac{128\pi}{5},\quad
  \mathcal{A}_1 = -192\sqrt2\arctan\sqrt2,\quad
  \mathcal{A}_2 = \tfrac{128\pi + 256\sqrt3}{3},\quad
  \mathcal{A}_3 = -72\sqrt2\arctan\sqrt2 - 16,
\end{equation}
with $\mathcal{A}_4 = 0$ and
$\mathcal{A}_5 = \tfrac{32}{15} + \tfrac{48}{5}\sqrt2\arctan\sqrt2$.  Hence,
for $0 \le Y \le \tfrac16$,
\begin{align}
  F^{(2,1)}(Y) \;=\;\;&
    \frac{128\sqrt3\,\pi}{5}\,Y^{3/2}
    \;-\; 432\sqrt2\arctan\!\sqrt2\;Y^{2}
    \;+\; \frac{384\sqrt3\,\pi + 2304}{5}\,Y^{5/2} \notag\\
  &-\; \bigl(72 + 324\sqrt2\arctan\!\sqrt2\bigr)Y^{3}
    \;+\; \frac{108 + 486\sqrt2\arctan\!\sqrt2}{5}\,Y^{4}.
  \label{eq:FB21}
\end{align}
The $Y^2$ coefficient is exactly $-36K_g$, matching the edge
formula $-\tfrac94\bigl(f(0)^4+f(1)^4\bigr)K_g$ that the first-order
slab expansion gives at $n = 4$ for parents with mass at a support
endpoint.

\subsection{Worked example: Beta\texorpdfstring{$(2,2)$}{(2,2)}}

For $f(x) = 6x(1-x)$ the degrees are $\{0,2,4,5,6,7,9\}$, all even
moments are pure $\mathbb{Q}\pi + \mathbb{Q}\sqrt3$,
\begin{equation}
  \mathcal{A}_0 = \tfrac{576\pi}{35},\quad
  \mathcal{A}_2 = -\tfrac{3456\pi}{35},\quad
  \mathcal{A}_4 = \tfrac{3456\pi}{25},\quad
  \mathcal{A}_6 = -\tfrac{16128\pi + 27648\sqrt3}{35},
\end{equation}
and the odd moments are
$\mathcal{A}_5 = \tfrac{864}{5} + \tfrac{3888}{5}\sqrt2\arctan\sqrt2$,
$\mathcal{A}_7 = \tfrac{3942}{35} + \tfrac{37179}{35}\sqrt2\arctan\sqrt2$,
$\mathcal{A}_9 = -\tfrac{1752}{175} - \tfrac{16524}{175}\sqrt2\arctan\sqrt2$.
Assembling \eqref{eq:poly-structure}, for $0 \le Y \le \tfrac16$,
\begin{align}
  F^{(2,2)}(Y) \;=\;\;&
   \frac{576\sqrt3\,\pi}{35}Y^{3/2}
   - \frac{31104\sqrt3\,\pi}{175}Y^{5/2}
   + \frac{93312\sqrt3\,\pi}{175}Y^{7/2}
   + \frac{8748 + 39366\sqrt2\arctan\sqrt2}{5}Y^{4} \notag\\
  &- \Bigl(\frac{20736\sqrt3\,\pi}{5} + \frac{746496}{35}\Bigr)Y^{9/2}
   + \frac{957906 + 9034497\sqrt2\arctan\sqrt2}{350}Y^{5} \notag\\
  &- \frac{106434 + 1003833\sqrt2\arctan\sqrt2}{175}Y^{6}.
  \label{eq:FB22}
\end{align}
There is no $Y^2$ term: an interior parent has $f(0) = f(1) = 0$ and
the edge coefficient vanishes.  The three half-integer head
coefficients agree with the general small-$Y$ expansion
$c_1 = 8\sqrt3\pi\int f^4$ and its successors.

\subsection{The upper regimes: full-support closure}
\label{sec:poly-upper}

The remaining probability mass---$1.13\%$ for Beta$(2,1)$, $0.28\%$
for Beta$(2,2)$---lives beyond $Y = 1/6$, where the ball crosses the
support zonotope and the cap and wall-cut geometry of
Sections~\ref{sec:derive-mid} and~\ref{sec:derive-high} switches on
with polynomial weights.  The radial split of \eqref{eq:master} at
$r = \min(R, 1/g)$ gives the CDF everywhere in one stroke:
\begin{equation}
  F(Y) \;=\; \sum_d \bigl[\mathcal{A}_d - \widehat C_d(s)\bigr]
     \frac{R^{\,d+3}}{d+3}
  \;+\; \sum_d \frac{\widehat D_d(s)}{d+3},
  \qquad s = \arccos\frac{1}{\sqrt2 R},
  \label{eq:poly-upper}
\end{equation}
where $\widehat C_d(s) = 24\int_0^s G_d(\psi)\sin\psi\dd\psi$ and
$\widehat D_d(s) = 24\int_0^s G_d(\psi)\,(\sqrt2\cos\psi)^{-(d+3)}
\sin\psi\dd\psi$ are moments of the azimuthal integrals
$G_d(\psi) = \int_{X(\psi)} \rho_d\dd\varphi$ of the weight over
the chamber's azimuth set $X(\psi)$---the same three-piece set as in
Appendix~\ref{app:trig}, with wall arcs of half-width $\beta(\psi)$
centered at azimuths $\pm\arctan\sqrt2$ from the chamber's central
azimuth, i.e.\ $\operatorname{arccot}\sqrt2$ from each corner.
Every $G_d$ is obtained by evaluating trigonometric polynomial
primitives at the arc endpoints, together with a possible term
proportional to the angular length of the arc; the endpoint angles
themselves are inverse trigonometric functions of algebraic
quantities.  For the two parents below both $\psi$-integrals were
carried out in closed form, so \eqref{eq:poly-upper} is a finite
closed form in every regime; for a general polynomial parent the
elementarity of the assembled $F$ is Theorem~\ref{thm:poly}, whose
proof gives a finite construction that does not pass through these
$\psi$-integrals.  No
matching constants are fitted anywhere: \eqref{eq:poly-upper} is a
single expression of the moving limit $s$, continuous across the
thresholds by construction.

For the middle regime $\tfrac16 \le Y \le \tfrac29$ the results are
compact.  For Beta$(2,1)$,
\begin{align}
  F^{(2,1)}(Y) \;=\;\;& \pi\Bigl(-\tfrac{19}{20\sqrt2} + 33\sqrt2\,Y
   - 128\sqrt3\,Y^{3/2} + 459\sqrt2\,Y^{2}
   - 384\sqrt3\,Y^{5/2} \notag\\
  &\qquad + 324\sqrt2\,Y^{3}
   - \tfrac{486\sqrt2}{5}\,Y^{4}\Bigr) \notag\\
  &+ \tfrac{18}{5}\,Y^{2}\Bigl(128\sqrt{Y}
   - 120\sqrt2\arctan\sqrt2 \notag\\
  &\qquad - 10\,\bigl(2 + 9\sqrt2\arctan\sqrt2\bigr)\,Y
   + 3\,\bigl(2 + 9\sqrt2\arctan\sqrt2\bigr)\,Y^{2}\Bigr),
  \label{eq:FB21mid}
\end{align}
and for Beta$(2,2)$,
\begin{align}
  F^{(2,2)}(Y) \;=\; \frac{3}{5600}\biggl[\,&\pi\Bigl(-1639\sqrt2
    + 85617\sqrt2\,Y - 153600\sqrt3\,Y^{3/2}
    - 629370\sqrt2\,Y^{2} \notag\\
  &\quad + 1658880\sqrt3\,Y^{5/2} \notag\\
  &\quad + 1247400\sqrt2\,Y^{3} - 4976640\sqrt3\,Y^{7/2}
    - 14696640\sqrt2\,Y^{4} + 38707200\sqrt3\,Y^{9/2} \notag\\
  &\quad - 48183984\sqrt2\,Y^{5} + 10707552\sqrt2\,Y^{6}\Bigr) \notag\\
  &- 3888\,Y^{4}\Bigl(10240\sqrt{Y}
    - 420\,\bigl(2 + 9\sqrt2\arctan\sqrt2\bigr) \notag\\
  &\qquad - 9\,\bigl(146 + 1377\sqrt2\arctan\sqrt2\bigr)\,Y
    + \bigl(292 + 2754\sqrt2\arctan\sqrt2\bigr)\,Y^{2}\Bigr)\biggr].
  \label{eq:FB22mid}
\end{align}
The upper-middle and high regimes are longer; the four expressions
are recorded in Appendix~\ref{app:betafull}.  Four certificates
anchor the printed forms.  The middle-regime forms above meet the
low-regime forms \eqref{eq:FB21}--\eqref{eq:FB22} at $Y = 1/6$ with
difference exactly zero in exact arithmetic.  The regime joints at
$2/9$ and $1/4$ are continuous by construction of
\eqref{eq:poly-upper}, and the printed tables of
Appendix~\ref{app:betafull} meet there to $10^{-50}$ in $60$-digit
evaluation.  The ceiling identity $F(1/3) = 1$ reduces to exactly
zero for both parents after the identities
$\arctan(2\sqrt2) = \pi - 2\arctan\sqrt2$,
$\arctan(4\sqrt2/7) = 4\arctan\sqrt2 - \pi$,
$\arccos\sqrt{2/3} = \pi/2 - \arctan\sqrt2$ and
$\operatorname{arcsec}\sqrt3 = \arctan\sqrt2$.  And the first two
moments of $S^2$ computed from the printed CDFs reproduce the exact
values $\sigma^2$ and $\sigma^4\bigl(1 + 2/3 + \kappa/4\bigr)$, with
$\kappa$ the excess kurtosis of the parent, to $30$ digits, a check
that involves every regime at once.  In particular every constant in the full-support
CDFs stays inside the uniform-case field
$\mathbb{Q}\bigl(\pi, \sqrt2, \sqrt3, \arctan\sqrt2\bigr)$.

\subsection{Validation}

Each $\mathcal{A}_d$ was verified by a hybrid of exact and numerical
steps: the symmetric part of the
integrand is evaluated in closed form as a Gaussian moment, and the
non-symmetric remainder, where it does not vanish identically, against
an independent $60$-digit Gauss--Legendre quadrature over the chamber
triangle, agreeing to all computed digits; four of the twelve moments
have zero remainder and are therefore exact outright.  The azimuthal
integrals $G_d$ were checked against adaptive quadrature over the exact
arc intervals to $10^{-13}$.  End to end, the assembled CDFs were
checked independently in the original $X$-coordinates by a
deterministic tensor-Gauss quadrature of
$\int f(x_1)f(x_2)f(x_3)\,[F_f(b)-F_f(a)]\dd x_{1:3}$, with $[a,b]$
the exact $x_4$-interval of the variance constraint.  At $Y = 0.05$ and
$0.13$ the largest observed absolute discrepancy is $5.83\eee{7}$, and
at $Y = 0.18$, $0.20$, $0.235$, $0.26$ and $0.30$, in the three upper
regimes, it is $\le 10^{-9}$, for both parents.  The $320$-to-$480$-node
refinement changes at $Y = 0.05$ are reported in Table~\ref{tab:betaquad}; they are convergence
diagnostics, not certified error bounds, and there the observed discrepancy exceeds the
refinement change for both Beta parents.  Monte Carlo at
$N = 4\times10^7$ draws agrees within
binomial noise at every tested $Y$ in every regime.  Figure~\ref{fig:poly-cdf} plots the three
exact laws on the full support.

\begin{table}[htbp]
\centering
\caption{Low-regime quadrature check at $Y=0.05$: absolute formula--quadrature discrepancy at
$480$ nodes per coordinate, against the $320\to480$-node refinement change. The refinement change
is a convergence diagnostic, not a certified error bound, and does not bound the discrepancy for
either parent.}
\label{tab:betaquad}
\begin{tabular}{lrrr}
\toprule
Parent & Formula--480-node discrepancy & $320\to480$ change & Ratio \\
\midrule
Beta$(2,1)$ & $2.58889\times10^{-7}$ & $9.56508\times10^{-8}$ & $2.7066$ \\
Beta$(2,2)$ & $5.82856\times10^{-7}$ & $4.90441\times10^{-7}$ & $1.1884$ \\
\bottomrule
\end{tabular}
\end{table}

\subsection{Provenance of the results}\label{sec:provenance}

Computer algebra participates in several derivations, so we record
what is established how.

\begin{center}
\renewcommand{\arraystretch}{1.2}
\small
\begin{tabular}{@{}ll>{\raggedright\arraybackslash}p{5.4cm}@{}}
\toprule
result & derivation & certificate \\
\midrule
critical radii (\S\ref{sec:geometry}) & projection geometry &
  hull distances, $10$ digits \\
low regime \eqref{eq:Flow}; $K_g$ \eqref{eq:Kg} & by hand & exact \\
middle regime \eqref{eq:Fmid} & by hand, cap integral &
  identity check $2\eee{14}$ \\
upper densities \eqref{eq:fR4a}--\eqref{eq:fR4b} &
  App.~\ref{app:trig} $+$ CAS &
  assembly $\equiv$ closed form, $8\eee{14}$ \\
$\Phi_a, \Phi_b$ & CAS &
  $60$-digit derivative residual $10^{-58}$ \\
joint constants at $2/9$, $1/4$ & CAS & symbolic \\
$m(\psi)$, wall structure & by hand &
  spherical sampling; tiling identity exact \\
moments $\mathcal{A}_d$ & divergence recursion, exact &
  exact where the remainder vanishes, else $60$-digit quadrature of the remainder \\
full-support Beta forms & \eqref{eq:poly-upper} $+$ CAS &
  exact joint at $1/6$ and ceiling \\
Beta joints at $2/9$, $1/4$ & by construction &
  printed tables agree to $10^{-50}$ \\
all printed CDFs & --- & moments of $S^2$ exact to $30$ digits \\
independent quadrature & original coordinates &
  $\le 5.8\eee{7}$; $\le 10^{-9}$ upper regimes \\
Monte Carlo & --- & binomial noise, $N = 2$--$4\times10^7$ \\
\bottomrule
\end{tabular}
\end{center}

\begin{figure}[t]
\centering
\includegraphics[width=\linewidth]{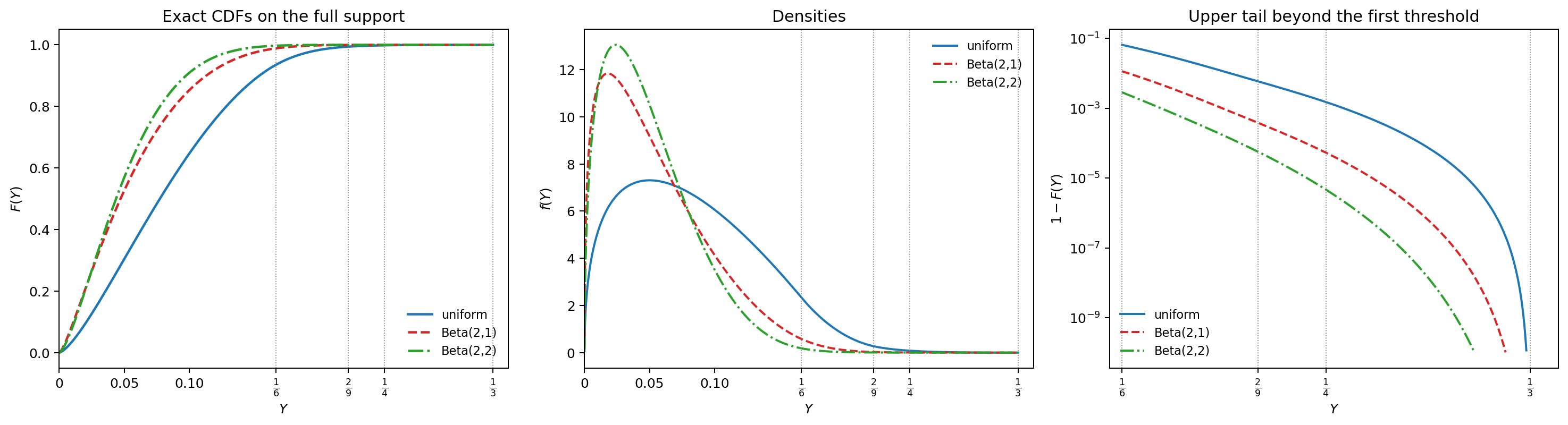}
\caption{Exact CDFs, densities and upper tails at $n = 4$ on the
full support for the uniform, Beta$(2,1)$, and Beta$(2,2)$ parents;
dotted verticals mark the parent-independent thresholds. The middle
panel's densities are central finite-difference approximations to
the exact closed-form CDFs ($h=10^{-6}$), not the analytic
derivatives of Corollary~\ref{cor:density}. The right
panel shows $1 - F$ on a logarithmic scale beyond $Y = 1/6$, down to
$10^{-10}$, below which evaluating $1 - F$ from the closed forms in
double precision is limited by cancellation.  Both Beta
parents accumulate variance mass faster than the uniform parent; as
$Y \to 0$ their leading coefficients stand to the uniform one in the
ratios $\int f^4 = 16/5$ and $72/35$, and the more concentrated
Beta$(2,2)$ law overtakes Beta$(2,1)$ early in the regime.  All
three curves are finite elementary expressions with constants in
$\mathbb{Q}\bigl(\pi,\sqrt2,\sqrt3,\arctan\sqrt2\bigr)$.}
\label{fig:poly-cdf}
\end{figure}

\section{Comparison with Royen's series}\label{sec:royen}

\citet{royen2007b,royen2008} represents the variance CDF for
any bounded continuous parent as a Fourier cosine series in
$x = (n-1)Y$,
\[
  F(Y) \;=\; 1 - \frac{\operatorname{E} Q}{q} \;-\; \frac{2}{\pi}
  \sum_{k \ge 1} \frac{\operatorname{Im}\hat f_Q(t_k)}{k}\,
  \cos(t_k x),
  \qquad t_k = \frac{k\pi}{q},
\]
where $Q = (n-1)S^2$, $q = \sup Q$, and each coefficient
$\hat f_Q(t_k)$ requires an oscillatory integral of an $n$-th power
of a characteristic-function factor.  The constant is Royen's
closed-form value of the $x$-free part of his series, $3/4$ for the
uniform parent at $n = 4$; summing that part term by term instead
leaves a truncation error that is the same at every $Y$ and decays
slowly in the number of terms $K$.  Instantiated at $n = 4$ for the
uniform parent and measured against Theorem~\ref{thm:main} at seven
points spread over the support, the series needs a median of $3$,
$9$ and $18$ terms to first reach $10^{-2}$, $10^{-3}$ and
$10^{-4}$, and $4$, $11$ and $29$ terms to stay within them for all
computed truncations $K \le 220$; at $K = 220$ its error is at most
$6.4\times10^{-7}$ at those points and $1.3\times10^{-5}$ on a grid
over $[0.01, 0.32]$.  These are observed maxima over the computed
truncations and the stated grid, not remainder bounds.  It is slower near the origin:
at $K = 220$ the partial sum still misses $F(0) = 0$ by
$1.8\times10^{-4}$, so accuracy that is uniform over the whole support
is limited by the endpoint.  Building the $220$ coefficients, each an
oscillatory quadrature, took about half a minute; evaluating the
cached series then costs microseconds per point.  The closed form is
exact everywhere, endpoints included, at the fixed cost of a handful
of inverse-trigonometric evaluations and with no construction step.
The series retains its structural advantages---uniformity in the
parent and in $n$---and remains the natural method where no closed
form is known.  It also improves with $n$.  With one implementation
across sample sizes and the same seven relative positions, measured
against the closed forms at $n = 3, 4$ and a numerical reference
evaluator at $n = 5$, the median number of terms to first reach $10^{-3}$ on the
uniform parent falls from $13$ at $n = 3$ to $9$ at $n = 4$ and $4$
at $n = 5$, and the endpoint error at $K = 220$ from
$3.3\times10^{-3}$ to $1.8\times10^{-4}$ to $8\times10^{-7}$: each
Fourier coefficient carries the $n$-th power of a characteristic
function, and higher powers decay faster.  The closed form's advantage
is therefore largest at the small $n$ where it exists, and the ledger
of Section~\ref{sec:ledger} records the comparison.

\section{Discussion}\label{sec:discussion}

\paragraph{The pattern across $n$.}
At $n = 3$ the contrast space is a plane: the cube projects to a
hexagon, the variance set is a disk, and the CDF closes elementarily
with one interior threshold \citep{n3paper}.  At $n = 4$ the contrast
space is $\RR^3$: the cube projects to the rhombic dodecahedron, the
variance set is a ball, and the CDF again closes elementarily, now
with three interior thresholds---one per boundary orbit of the richer
polytope.  The same reduction at general $n$ expresses the CDF
through a solid-angle integral over $S^{n-2}$ cut by the braid
arrangement $A_{n-1}$, with $n(n-1)$ congruent cut-caps.  At $n = 5$
the low regime again recovers \citet{aliRichards1975}
with $b(5) = K_5/\sqrt5$.  Whether the closure persists beyond the low
regime for $n \ge 5$ is a separate question: the reduction extends,
but our proofs of elementary closure at $n = 3$ and $n = 4$ use that
the cut objects are points and arcs, and one dimension up they are
not.  We
make no claim here about larger $n$.  The present paper is the most
complete concrete instance of the general reduction.

\subsection{From three observations to four: a ledger}\label{sec:ledger}

It is worth recording, side by side, what the passage from $n = 3$ to
$n = 4$ preserved, what it cost, what was not carried over, and what
it says about larger $n$.  Table~\ref{tab:ledger} is the ledger.

\begin{table}[t]
\caption{Three observations against four, uniform parent.  The last
row is the median number of terms Royen's universal series,
assembled with its exact constant, needs to first reach absolute
accuracy $10^{-3}$ at seven points spread over the support, with one
implementation for both sample sizes.}
\label{tab:ledger}
\begin{center}
\begin{tabular}{@{}lll@{}}
\toprule
 & $n = 3$ & $n = 4$ \\
\midrule
contrast space & disk in a hexagon & ball in a rhombic dodecahedron \\
interior thresholds & $1/4$ (edges) & $1/6$, $2/9$, $1/4$ (faces, edges, corner vertices) \\
ceiling & $1/3$ (vertices) & $1/3$ (axial vertices) \\
regimes & $2$ & $4$ \\
constants & $\mathbb{Q}(\pi, \sqrt3)$ & $\mathbb{Q}(\pi, \sqrt2, \sqrt3, \arctan\sqrt2)$ \\
ceiling law, $1 - F(\tfrac13 - \delta)$ & $\tfrac{27}{2}\,\delta^3$ & $\tfrac{81}{4}\,\delta^4$ \\
Royen terms for $10^{-3}$ & $13$ & $9$ \\
\bottomrule
\end{tabular}
\end{center}
\end{table}

\paragraph{What extended, and at what cost.}
Everything structural extended.  The reduction to a ball cut by the
projected cube, the reading of thresholds as boundary-orbit distances,
the polynomial closure theorem for the whole class, and the exactness
of every constant all carry over, and at $n = 4$ the closure proof
names its mechanism: by the divergence theorem every moment reduces
to trigonometric integrals along circular arcs, spherical areas given
by the Gauss--Bonnet formula, and planar integrals over regions
bounded by segments and arcs.  The price is size.  Two regimes became
four, one variable angle became three, since the angle functions of
Theorem~\ref{thm:main} reduce to $A_1$, $A_2$ and $A_4$ through
$W = A_1$, $A_3 = \pi/2 - A_2$ and $B_\pm = A_4 \pm A_2$, the
constant field acquired $\sqrt2$ and $\arctan\sqrt2$, and the upper
regimes required the machine-generated Appendix~\ref{app:beta-upper}
rather than a display.  Regularity does not deteriorate: the density is $C^1$ in both
cases, and the only half-power singularity at $n = 4$, at $2/9$, has
the same origin as the single $n = 3$ bifurcation at $1/4$.  Both are
square-root onsets of an arc measure---at $n = 3$ the arcs of the
variance circle beyond the hexagon's edges, at $n = 4$ the wall cuts
of the caps---and produce a half-power divergence $(Y - Y_*)^{-1/2}$,
in the third derivative of $F$ at $n = 3$ and in the fourth at
$n = 4$, one derivative later because the arc measure is integrated
once more.  The other two $n = 4$ thresholds, the appearance of a cap
at $1/6$ and the merger of wall cuts at $1/4$, produce jumps in a
derivative rather than half powers.

\paragraph{What was not carried over.}
The $n = 3$ paper's closure hierarchy below the polynomial row---the
proof that a rational parent forces a dilogarithm, the genus-two
radial obstruction of the arcsine parent, the identification of the
Bessel and error-function rows---has no $n = 4$ counterpart in this
paper.  We expect each to persist, since the fiber integral that
creates each obstruction is unchanged, but we have not proved it, and
an extra integration can cancel transcendental functions as well as
create them; these rows remain conjectures at $n = 4$.  The singular-parent program transfers in its analytic
part and not yet in its numerical part: the endpoint laws for the
arcsine parent at $n = 4$ follow from the same corner and tube
arguments, with the exponents $n/4 = 1$ at the origin and
$2a + 2b = 2$ at the ceiling and the exact ceiling coefficient
$27/\pi^2 = 2.7357$ from the Dirichlet integral over the simplex at
an axial vertex, both analytic consequences of those arguments rather
than numerically checked laws at $n = 4$, but the
six-term approximation that made the arcsine parent computable at
$n = 3$ has not been built at $n = 4$.  It could be, by the same
subtraction of the two endpoint laws, and we regard it as the natural
next step for singular parents.

\paragraph{What the closure depends on.}
Two elementary cases in a row invite the guess that the pattern
continues.  The proof of Theorem~\ref{thm:poly} says what the guess
would have to rest on.  At $n = 3$ the cut objects on the circle are
points; at $n = 4$ they are arcs of circles on $S^2$, and the proof
uses exactly that: along a circular arc every boundary term is a
trigonometric polynomial, and the area of a region bounded by such
arcs is elementary by the Gauss--Bonnet formula.  One dimension up
the cut objects are spherical polygons, and neither step transfers as
stated.  The closure at $n = 3$ and $n = 4$ is therefore a statement
about those two geometries, and we make no claim for larger $n$; the
ledger shows the cost rising steeply even where it holds.

\paragraph{Royen's series and the normal limit.}
Two things gain ground on the closed form as $n$ grows, and they
should be separated.  The first is Royen's universal series, whose
term counts fall with $n$ at every accuracy from $10^{-3}$ down, from
$13$ terms for $10^{-3}$ at $n = 3$ to $9$ at $n = 4$ and $4$ at
$n = 5$ (Section~\ref{sec:royen}).
The coefficients of that series decay like a power of the
characteristic function, and higher powers decay faster, so the series
improves with $n$ at the same time as the closed form becomes harder
to obtain.  The second is the central limit theorem: $S^2$ is
asymptotically normal, and one might expect the exact law to matter
less as $n$ grows.  Table~\ref{tab:approx} tests that expectation
against the exact CDFs at $n = 3$ through $6$.  The expectation is
correct for a normal approximation that carries the right variance,
$\operatorname{Var} S^2 = \sigma^4\bigl[2/(n-1) + \kappa/n\bigr]$ with
$\kappa = -6/5$ the excess kurtosis of the uniform parent: its
grid-maximum error (Table~\ref{tab:approx}) falls from $0.098$ to $0.026$ and its error in the
$99$th percentile from $8\%$ to $4\%$.  It fails for the
normal-theory recipe $(n-1)S^2/\sigma^2 \sim \chi^2_{n-1}$, which is
what acceptance-testing practice uses.  That recipe assumes
$\kappa = 0$, so the actual variance of $S^2$ divided by the recipe's
is $1 + \kappa(n-1)/(2n)$, which tends to $2/5$, and its grid-maximum
error \emph{grows}, from $0.064$ to $0.102$, toward a nonzero limit
of $0.109$.  Its $99$th-percentile error does shrink, but slowly: from
$52\%$ to $42\%$ over this range, and like $1.21\,n^{-1/2}$ in the
limit.  The exact closed forms therefore matter most at the small $n$
where the corrected approximation has not yet arrived and the recipe
in use is far off in both senses; what they buy is diminished not by
the normal limit but by the availability of the kurtosis correction,
which is a different question and has a different answer for every
parent.

\begin{table}[t]
\caption{Exact CDF of $S^2$ for the uniform parent against three
approximations, $n = 3$ to $6$: the normal-theory recipe
$\chi^2_{n-1}$, a scaled $\chi^2$ with mean and variance matched to
the exact $S^2$, and the kurtosis-corrected normal.  Maximum absolute
error over an endpoint-inclusive grid of $4002$ points spanning
$[0, Y_{\max}]$, and relative error in the $99$th percentile.  These
are maxima over that grid rather than proven suprema; the largest
entry, $0.098$ for the normal recipe at $n = 3$, is attained at the
endpoint $Y = 0$, where the error is $\Phi(-\sqrt{5/3}) = 0.09835$.
Exact CDFs: closed forms at $n = 3, 4$; at $n = 5, 6$, piecewise
Chebyshev numerical reference evaluators on the exact regime
partition, whose observed maximum error against an independent
deterministic reference was below $10^{-8}$.  The $n = 5, 6$ rows are
therefore conditional on those reference evaluators, whose
construction is not reproduced here.}
\label{tab:approx}
\begin{center}
\small
\begin{tabular}{@{}lcccccc@{}}
\toprule
 & \multicolumn{3}{c}{$\max_Y |F - \widetilde F|$} & \multicolumn{3}{c}{$99$th percentile, relative error} \\
\cmidrule(lr){2-4}\cmidrule(lr){5-7}
$n$ & $\chi^2_{n-1}$ & matched $\chi^2$ & normal & $\chi^2_{n-1}$ & matched $\chi^2$ & normal \\
\midrule
$3$ & $0.064$ & $0.066$ & $0.098$ & $+52\%$ & $+19\%$ & $-8\%$ \\
$4$ & $0.083$ & $0.052$ & $0.055$ & $+50\%$ & $+15\%$ & $-5\%$ \\
$5$ & $0.095$ & $0.046$ & $0.036$ & $+45\%$ & $+11\%$ & $-5\%$ \\
$6$ & $0.102$ & $0.042$ & $0.026$ & $+42\%$ & $+9\%$ & $-4\%$ \\
\bottomrule
\end{tabular}
\end{center}
\end{table}

\paragraph{What was learned.}
Three things, each visible only from the pair of cases.  First, the
thresholds are geometric objects: they are the critical values of the
distance from the diagonal on the boundary strata of the projected
cube, one per boundary orbit.  Second, in both cases the functions
that appear are dictated by the cut objects rather than by the
dimension of the problem: at $n = 3$ the limits of integration are
angles, at $n = 4$ the boundary terms are trigonometric integrals
along arcs, and both give algebraic and inverse trigonometric
functions.  We offer this as a description of the two cases, not as a
rule for larger $n$.  Third, polynomial closure at $n = 4$ is a
statement about polynomial moments of a ball clipped by a polytope,
proved by reducing each moment to its boundary; it holds for the
whole polynomial class at $n = 3$ and $n = 4$ and should not be
assumed beyond.

\paragraph{From variance to standard deviation.}
$\Pr(S \le s) = F(s^2)$, so the standard-deviation CDF inherits the
closed form with thresholds at the square roots of the four $Y$
values; nothing new is needed.

\paragraph{Application to acceptance testing.}
Content-uniformity criteria evaluate events built on small-sample
standard deviations of bounded measurements \citep{usp905,williamsHauck2002}.  For four-unit designs, Theorem~\ref{thm:main}
turns the variance-channel probability into an exact elementary
evaluation with no Monte Carlo and no series truncation, and the
$Y^{3/2}$ small-$Y$ law with its exact second-order coefficient
$-\tfrac92 K_g$ quantifies how sharply probability accumulates near
zero variance for a flat parent.  What the exact laws supply is a controlled small-sample
benchmark, against which variance calculations and general
acceptance-probability algorithms can be checked where their own error
is otherwise unmeasurable. Table~\ref{tab:approx} sets the scale: applying the
normal-model variance law to the stated uniform parent misplaces the
$99$th percentile of $S^2$ by $50\%$ at $n = 4$, because that law's
variance omits the parent's kurtosis. The figure is a statement about
the two models compared, not a measured error in any fielded
procedure; the harmonized $\langle905\rangle$ text prescribes ten- and
thirty-unit stages, and we make no claim about a four-unit one.  That quantile error decays only like $n^{-1/2}$,
and the recipe's CDF error does not decay at all.  A closed form
removes the error entirely at the sample sizes it covers.  The companion paper \citep{n3paper}
develops the statistical applications systematically; we record here
only that the $n = 4$ uniform case, to our knowledge previously
beyond closed-form reach, is now within it.

\section*{Acknowledgments}

The author thanks Kevin T.~Chu and Ron Yu, co-authors of the companion paper on three observations
\citep{n3paper}, whose geometric formulation and exact uniform-parent derivations this work
builds on, and Nicole Tin, Somdatta Bhattacharya and Chiajen Lai for extensive
conversation about the pharmaceutical applications and in-silico tools that motivated this
work.  The author also thanks Mandeep Chauhan, Mike Salem, Ash Baghai, Wei-Ting Liu, Shikhar Mohan, Adedayo
Sanni, Chris Remmers, Caroline Branch, Kiara Cui, Dana Caulder, Joanna Koziara, Victor Rucker,
Jennifer Hedborn and Henry Morrison for motivating discussions at Gilead Sciences.

\section*{Statements and Declarations}

\paragraph{Funding.} No external funding was received for this work.

\paragraph{Competing interests.} The author is an employee of Gilead Sciences and has no other
competing interests to declare.

\paragraph{Author contributions.} R.~Osan is the sole author and carried out all of the work.

\paragraph{Data availability.} No data were generated or analyzed. Every parent distribution in
this paper is an analytic law stated in the text.

\paragraph{Code availability.} Every formula is stated in full and can be implemented from the
text. The scripts behind the derivations and validations --- the geometric verification of the
critical radii and the four-regime assembly with its continuity, ceiling and Monte Carlo checks;
the computer-algebra derivation and certification of $\Phi_a,\Phi_b$ and the matching constants;
the verification of the wall-arc geometry, cell measure and moment antiderivatives; the exact
spherical-triangle moment engine of \eqref{eq:div-recursion} with its $60$-digit quadrature
cross-checks; the full-support Beta pipeline of Section~\ref{sec:poly-upper}; exact and
$60$-digit certificates of every printed formula, which fail when any printed coefficient is
altered; an independent quadrature in the original coordinates for all three parents in every
regime; Royen's series with its exact constant; the comparison of Table~\ref{tab:approx}; and the
figure generation --- are available from the author on request and will be released with the
published article.

\paragraph{Use of generative AI.} Claude models (Anthropic) assisted with derivations, code
implementation and manuscript drafting. All outputs were verified by the author against
independent symbolic and numerical computations (Python, MATLAB).

\appendix

\section{The upper-regime densities, spelled out}\label{app:trig}

This appendix derives \eqref{eq:fR4a} and \eqref{eq:fR4b} from the
spherical geometry, with every intermediate object stated explicitly.
Each claim below was additionally machine-verified; the assembled
one-dimensional integral of \eqref{eq:app-assembly} reproduces the
closed densities of Corollary~\ref{cor:density} to $10^{-14}$ at spot
points in every regime.

\subsection{From the master integral to one cap}

Start from \eqref{eq:density} and split off the full-sphere part:
\begin{equation}
  f(Y) \;=\; 3R\bigl(4\pi - R\,K_g\bigr) \;+\; 3R\,\mathcal E(R),
  \qquad
  \mathcal E(R) \;=\; \int_{\{g \ge 1/R\}}
      \bigl(R\,g(\omega) - 1\bigr)\dd\Omega ,
  \label{eq:app-excess}
\end{equation}
so all regime structure lives in the excess integral $\mathcal E$
over the excision set.  The six great circles $\{c_i = c_k\}$, the
$A_3$ braid arrangement on $S^2$ \citep{humphreys1990}, cut the
sphere into $24$ chambers
on which the identities of $\arg\max_i c_i$ and $\arg\min_i c_i$ are
constant; call the union of chambers with $\arg\max = i$,
$\arg\min = j$ the \emph{cell} of the ordered pair $(i,j)$.  On cell
$(i,j)$ the range gauge is the single linear functional
$g = b_{ij}\cdot\omega = \sqrt2\cos\psi$, with $\psi$ the colatitude
from the unit cap axis $u_{ij} = b_{ij}/\sqrt2$.  The coordinate
permutations act on contrast space as the reflection group
$T_d \cong S_4$, carrying cells onto cells isometrically, so the twelve cell contributions to
$\mathcal E$ are equal:
\begin{equation}
  \mathcal E(R)
   \;=\; 12 \int_0^{\psi_R}
     \bigl(\sqrt2\,R\cos\psi - 1\bigr)\, m(\psi)\,\sin\psi \dd\psi,
  \qquad
  \psi_R = \arccos\!\frac{1}{\sqrt2\,R},
  \label{eq:app-assembly}
\end{equation}
where $m(\psi)$ is the azimuthal measure of cell $(1,2)$ at
colatitude $\psi$ from its axis.  The excision is nonempty once
$\sqrt2 R > 1$, i.e.\ $Y > 1/6$: the cap onset of
Section~\ref{sec:derive-mid}.

\subsection{The four walls of a cap}

The roots adjacent to $b_{12}$ are the four sharing one index:
$b_{13}, b_{14}, b_{32}, b_{42}$, each at $60^\circ$ from $b_{12}$
since $b_{12}\cdot b_{1k} = 1 = \tfrac12\,|b|^2$.  Their azimuths
around the axis $u_{12}$ are not equally spaced: computing the
projections into the plane perpendicular to $u_{12}$ gives the
alternating gaps
\begin{equation}
  \arccos\tfrac13 \;\approx\; 70.5288^\circ
  \qquad\text{and}\qquad
  \pi - \arccos\tfrac13 \;\approx\; 109.4712^\circ,
  \label{eq:app-gaps}
\end{equation}
the vertex figure of the cuboctahedron being a rectangle, not a
square.  The wall between cells $(1,2)$ and $(1,3)$ is the great
circle $\{c_2 = c_3\}$; a point of the cap at colatitude $\psi$ and
azimuth $\phi$, measured from the neighbor's azimuth, lies beyond
the wall exactly when $\omega\cdot b_{12} < \omega\cdot b_{13}$.
Writing $\omega = \cos\psi\,u + \sin\psi\,(\cos\phi\,e_1 +
\sin\phi\,e_2)$ with $e_1$ toward the neighbor, and using
$u_{13} = \tfrac12 u_{12} + \tfrac{\sqrt3}{2}e_1$:
\begin{equation}
  \cos\psi \;<\; \tfrac12\cos\psi + \tfrac{\sqrt3}{2}\sin\psi\cos\phi
  \quad\Longleftrightarrow\quad
  \cos\phi \;>\; \frac{\cot\psi}{\sqrt3},
\end{equation}
an azimuthal arc of half-width
\begin{equation}
  \beta(\psi) \;=\; \arccos\!\frac{\cot\psi}{\sqrt3},
  \label{eq:app-beta}
\end{equation}
nonempty exactly for $\psi > \pi/6$.  Differentiating
\eqref{eq:app-beta} exposes the advertised radical:
\begin{equation}
  \beta'(\psi) \;=\; \frac{1}{\sin\psi\,\sqrt{3 - 4\cos^2\psi}} ,
  \qquad
  3 - 4\cos^2\psi \;=\; 4\sin^2\psi - 1 .
  \label{eq:app-betaprime}
\end{equation}

\subsection{The azimuthal cell measure}\label{app:azimuthal}

Subtracting the four wall arcs from the full circle, and accounting
for their merging, gives a three-piece elementary formula:
\begin{equation}
  m(\psi) \;=\;
  \begin{cases}
    2\pi, & 0 \le \psi \le \tfrac{\pi}{6},\\[3pt]
    2\pi - 8\,\beta(\psi), &
      \tfrac{\pi}{6} \le \psi \le \operatorname{arccot}\sqrt2,\\[3pt]
    2\pi - 4\,\beta(\psi) - 2\arccos\tfrac13, &
      \operatorname{arccot}\sqrt2 \le \psi \le \tfrac{\pi}{4}.
  \end{cases}
  \label{eq:app-m}
\end{equation}
The first case boundary is the onset of the wall arcs and the second
an adjacent-arc merger; a third merger closes the cell at the top of
the range:
\begin{itemize}[leftmargin=1.6em, itemsep=2pt]
\item \emph{Arc onset.}  By \eqref{eq:app-beta}, $\beta(\pi/6) =
  \arccos 1 = 0$: below $\psi = \pi/6$ the cap meets no wall and the
  full circle survives, and at $\psi = \pi/6$ the four wall arcs
  appear with zero width.  This is the first case boundary of
  \eqref{eq:app-m}; no arcs have merged there.
\item \emph{Adjacent merge.}  Arcs across a narrow gap
  \eqref{eq:app-gaps} touch when $2\beta(\psi) = \arccos\tfrac13$;
  by the half-angle identity $\arccos\tfrac13 =
  2\operatorname{arccot}\sqrt2$ this solves to $\psi =
  \operatorname{arccot}\sqrt2 \approx 35.2644^\circ$.  Through
  $\psi_R$ this is exactly $Y = 1/4$, the corner-vertex threshold:
  past it the two merged pairs contribute their fixed union
  $2\beta + \arccos\tfrac13$ each, which is the third line of
  \eqref{eq:app-m}.
\item \emph{Opposite merge and tiling.}  Arcs across a wide gap touch
  when $2\beta(\psi) = \pi - \arccos\tfrac13$, which solves to $\psi
  = \pi/4$---exactly the largest colatitude the cap ever attains,
  reached at $Y = 1/3$.  There the cell closes up:
  \begin{equation}
    m\bigl(\tfrac{\pi}{4}\bigr)
    \;=\; 2\pi - 4\arccos\tfrac{1}{\sqrt3} - 2\arccos\tfrac13
    \;=\; 0,
    \label{eq:app-tiling}
  \end{equation}
  an exact identity.  This is the mechanism forcing
  $f(1/3^-) = 0$: at the ceiling the twelve cell-restricted caps,
  each cap intersected with its own cell, tile the entire sphere.
\end{itemize}
Formula \eqref{eq:app-m} was verified against direct spherical
sampling of the cell condition at eleven colatitudes spanning all
three pieces, with agreement at the $10^{-5}$ sampling resolution;
Figure~\ref{fig:atlas-mpsi} overlays the two.

\subsection{Moment integrals and assembly}

Substituting \eqref{eq:app-m} into \eqref{eq:app-assembly} reduces
$\mathcal E$ to elementary cap moments plus the two $\beta$-moments
\begin{equation}
  M_0(\psi) = \int \beta(\psi)\sin\psi\dd\psi,
  \qquad
  M_1(\psi) = \int \beta(\psi)\sin\psi\cos\psi\dd\psi,
\end{equation}
whose antiderivatives are elementary, carrying only the radical of
\eqref{eq:app-betaprime}:
\begin{align}
  M_0(\psi) &= -\beta(\psi)\cos\psi
     + \arctan\sqrt{4\sin^2\psi - 1},
  \label{eq:app-M0}\\
  M_1(\psi) &= \tfrac14\Bigl[-\beta(\psi)\cos2\psi
     - \arctan\frac{\cos\psi}{\sqrt{4\sin^2\psi-1}}
     + \arctan\frac{2\cos\psi}{\sqrt{4\sin^2\psi-1}}\Bigr].
  \label{eq:app-M1}
\end{align}
Both were certified by symbolic differentiation, which returns
exactly $\beta\sin\psi$ and $\beta\sin\psi\cos\psi$, and by
quadrature to $10^{-12}$ and $10^{-16}$ respectively.  Evaluating at
the moving limit $\psi_R$ converts the trigonometry of $\psi$ into
the algebra of $Y$: from $\cos\psi_R = 1/(\sqrt2 R)$,
\begin{equation}
  4\sin^2\psi_R - 1 = 3 - \frac{2}{3Y},
  \qquad
  \beta(\psi_R) = \arccos\frac{1}{\sqrt{18Y-3}},
\end{equation}
so $\arctan\sqrt{4\sin^2\psi_R - 1} = A_4$ and $\beta(\psi_R) = A_3$
in the notation of \eqref{eq:angles}, while the two arctangents of
\eqref{eq:app-M1} evaluated at $\psi_R$ are exactly $A_2$ and
$A_1$, since the squared arguments reduce to $1/(18Y-4)$ and
\begin{samepage}
$2/(9Y-2)$.  The fixed limits $\pi/6$,
$\operatorname{arccot}\sqrt2$, and $\pi/4$ contribute the constants
$\arccos\tfrac13$, $\arctan\sqrt2$, and $\arctan(\sqrt2/5)$
appearing in $\kappa_0$ and $\kappa_1$.  Collecting the terms of
\eqref{eq:app-assembly} for $\psi_R \le \operatorname{arccot}\sqrt2$
gives \eqref{eq:fR4a}; for $\psi_R > \operatorname{arccot}\sqrt2$ it
gives \eqref{eq:fR4b}.  As an end-to-end certificate, the numerical
evaluation of \eqref{eq:app-excess}--\eqref{eq:app-assembly} with
\eqref{eq:app-m} matches the closed densities at seven spot points
across the middle and upper regimes with differences at most
$8\times10^{-14}$, and integrating the same representation in $Y$
recovers the antiderivatives $\Phi_a, \Phi_b$ of
Theorem~\ref{thm:main}.

\section{Full-support closed forms for the Beta parents}
\label{app:beta-upper}
\label{app:betafull}

This appendix records the upper-middle ($\tfrac29 \le Y \le
\tfrac14$) and high ($\tfrac14 \le Y \le \tfrac13$) regime closed
forms assembled by \eqref{eq:poly-upper}; the low and middle regimes
appear in Section~\ref{sec:poly}.  Each expression is presented as a table of coefficients: with the
basis functions $\varphi$ of the left column,
\begin{equation*}
  F(Y) \;=\; \sum_{\varphi} \varphi \cdot P_\varphi(Y),
\end{equation*}
where the row $\varphi = 1$ carries the purely algebraic part.  The
tables are machine-generated from the certified computer-algebra
assembly and re-validated numerically after conversion; the
machine-readable originals accompany the manuscript.  Notation
matches the uniform case: $A_1, \dots, A_4$ are the angles
\eqref{eq:angles}, $B_\pm = B_\pm(Y)$ the arctangents
\eqref{eq:Bpm}, and
$W = W(Y) = \arctan\bigl(\sqrt2/\sqrt{9Y-2}\bigr)$, the angle
appearing in the uniform antiderivative \eqref{eq:Phia}.
\needspace{6\baselineskip}
Beta$(2,1)$, upper-middle regime:
\nopagebreak
\begin{center}\small
\renewcommand{\arraystretch}{1.35}
\begin{tabular}{>{$}l<{$}>{$}l<{$}}
\toprule
\varphi & \text{coefficient } P_\varphi(Y) \text{ in } F^{(2,1)},\ \tfrac29 \le Y \le \tfrac14 \\
\midrule
1 & \frac{2304 Y^{\frac{5}{2}}}{5} - 72 Y^{3} +\frac{108 Y^{4}}{5} \\
\pi & \frac{19 \sqrt{2}}{40} +33 \sqrt{2} Y - 128 \sqrt{3} Y^{\frac{3}{2}} +27 \sqrt{2} Y^{2} \\
 & \quad - 384 \sqrt{3} Y^{\frac{5}{2}} \\
\sqrt{9Y-2} & \frac{13}{9} - \frac{34 Y}{15} - \frac{1524 Y^{2}}{5} \\
\arctan\sqrt2 & - 432 \sqrt{2} Y^{2} - 324 \sqrt{2} Y^{3} +\frac{486 \sqrt{2} Y^{4}}{5} \\
A_3 & - 132 \sqrt{2} Y - 108 \sqrt{2} Y^{2} \\
A_4 & \frac{3072 \sqrt{3} Y^{\frac{3}{2}}}{5} \\
W & - 864 \sqrt{2} Y^{2} - 648 \sqrt{2} Y^{3} +\frac{972 \sqrt{2} Y^{4}}{5} \\
B_+ & - \frac{19 \sqrt{2}}{20} +864 \sqrt{2} Y^{2} +\frac{4608 \sqrt{3} Y^{\frac{5}{2}}}{5} +648 \sqrt{2} Y^{3} \\
 & \quad - \frac{972 \sqrt{2} Y^{4}}{5} \\
B_- & \frac{19 \sqrt{2}}{20} - 864 \sqrt{2} Y^{2} +\frac{4608 \sqrt{3} Y^{\frac{5}{2}}}{5} - 648 \sqrt{2} Y^{3} \\
 & \quad +\frac{972 \sqrt{2} Y^{4}}{5} \\
\bottomrule
\end{tabular}
\end{center}
\needspace{6\baselineskip}
\end{samepage}
Beta$(2,1)$, high regime:
\begin{samepage}
\nopagebreak
\begin{center}\small
\renewcommand{\arraystretch}{1.35}
\begin{tabular}{>{$}l<{$}>{$}l<{$}}
\toprule
\varphi & \text{coefficient } P_\varphi(Y) \text{ in } F^{(2,1)},\ \tfrac14 \le Y \le \tfrac13 \\
\midrule
1 & \frac{11}{18} - \frac{20 Y}{3} +108 Y^{2} +216 Y^{3} \\
 & \quad - \frac{324 Y^{4}}{5} \\
\pi & \frac{19 \sqrt{2}}{40} - \frac{384 \sqrt{3} Y^{\frac{3}{2}}}{5} - 432 \sqrt{2} Y^{2} - 384 \sqrt{3} Y^{\frac{5}{2}} \\
 & \quad - 324 \sqrt{2} Y^{3} +\frac{486 \sqrt{2} Y^{4}}{5} \\
\sqrt{9Y-2} & \frac{13}{18} - \frac{17 Y}{15} - \frac{762 Y^{2}}{5} \\
\arctan\sqrt2 & 66 \sqrt{2} Y +486 \sqrt{2} Y^{2} +324 \sqrt{2} Y^{3} - \frac{486 \sqrt{2} Y^{4}}{5} \\
A_3 & - 66 \sqrt{2} Y - 54 \sqrt{2} Y^{2} \\
A_4 & \frac{1536 \sqrt{3} Y^{\frac{3}{2}}}{5} \\
W & - 432 \sqrt{2} Y^{2} - 324 \sqrt{2} Y^{3} +\frac{486 \sqrt{2} Y^{4}}{5} \\
B_+ & - \frac{19 \sqrt{2}}{40} +432 \sqrt{2} Y^{2} +\frac{2304 \sqrt{3} Y^{\frac{5}{2}}}{5} +324 \sqrt{2} Y^{3} \\
 & \quad - \frac{486 \sqrt{2} Y^{4}}{5} \\
B_- & \frac{19 \sqrt{2}}{40} - 432 \sqrt{2} Y^{2} +\frac{2304 \sqrt{3} Y^{\frac{5}{2}}}{5} - 324 \sqrt{2} Y^{3} \\
 & \quad +\frac{486 \sqrt{2} Y^{4}}{5} \\
\arctan(4\sqrt2-3\sqrt3) & - \frac{19 \sqrt{2}}{40} +432 \sqrt{2} Y^{2} - \frac{2304 \sqrt{3} Y^{\frac{5}{2}}}{5} +324 \sqrt{2} Y^{3} \\
 & \quad - \frac{486 \sqrt{2} Y^{4}}{5} \\
\arctan(4\sqrt2+3\sqrt3) & - \frac{19 \sqrt{2}}{40} +432 \sqrt{2} Y^{2} +\frac{2304 \sqrt{3} Y^{\frac{5}{2}}}{5} +324 \sqrt{2} Y^{3} \\
 & \quad - \frac{486 \sqrt{2} Y^{4}}{5} \\
\bottomrule
\end{tabular}
\end{center}
\needspace{6\baselineskip}
\end{samepage}
Beta$(2,2)$, upper-middle regime:
\begin{samepage}
\nopagebreak
\begin{center}\small
\renewcommand{\arraystretch}{1.35}
\begin{tabular}{>{$}l<{$}>{$}l<{$}}
\toprule
\varphi & \text{coefficient } P_\varphi(Y) \text{ in } F^{(2,2)},\ \tfrac29 \le Y \le \tfrac14 \\
\midrule
1 & \frac{8748 Y^{4}}{5} - \frac{746496 Y^{\frac{9}{2}}}{35} +\frac{478953 Y^{5}}{175} \\
 & \quad - \frac{106434 Y^{6}}{175} \\
\pi & \frac{4917 \sqrt{2}}{5600} +\frac{36693 \sqrt{2} Y}{800} - \frac{576 \sqrt{3} Y^{\frac{3}{2}}}{7} \\
 & \quad - \frac{26973 \sqrt{2} Y^{2}}{80} +\frac{31104 \sqrt{3} Y^{\frac{5}{2}}}{35} +\frac{2673 \sqrt{2} Y^{3}}{4} \\
 & \quad - \frac{93312 \sqrt{3} Y^{\frac{7}{2}}}{35} +20736 \sqrt{3} Y^{\frac{9}{2}} \\
\sqrt{9Y-2} & \frac{26629}{18900} +\frac{11551 Y}{2100} - \frac{1251 Y^{2}}{175} \\
 & \quad - \frac{102708 Y^{3}}{175} +\frac{2463534 Y^{4}}{175} \\
\arctan\sqrt2 & \frac{39366 \sqrt{2} Y^{4}}{5} +\frac{9034497 \sqrt{2} Y^{5}}{350} - \frac{1003833 \sqrt{2} Y^{6}}{175} \\
A_3 & - \frac{36693 \sqrt{2} Y}{200} +\frac{26973 \sqrt{2} Y^{2}}{20} - 2673 \sqrt{2} Y^{3} \\
A_4 & \frac{13824 \sqrt{3} Y^{\frac{3}{2}}}{35} - \frac{746496 \sqrt{3} Y^{\frac{5}{2}}}{175} \\
W & \frac{78732 \sqrt{2} Y^{4}}{5} +\frac{9034497 \sqrt{2} Y^{5}}{175} - \frac{2007666 \sqrt{2} Y^{6}}{175} \\
B_+ & - \frac{4917 \sqrt{2}}{2800} +\frac{1119744 \sqrt{3} Y^{\frac{7}{2}}}{175} - \frac{78732 \sqrt{2} Y^{4}}{5} \\
 & \quad - \frac{248832 \sqrt{3} Y^{\frac{9}{2}}}{5} - \frac{9034497 \sqrt{2} Y^{5}}{175} +\frac{2007666 \sqrt{2} Y^{6}}{175} \\
B_- & \frac{4917 \sqrt{2}}{2800} +\frac{1119744 \sqrt{3} Y^{\frac{7}{2}}}{175} +\frac{78732 \sqrt{2} Y^{4}}{5} \\
 & \quad - \frac{248832 \sqrt{3} Y^{\frac{9}{2}}}{5} +\frac{9034497 \sqrt{2} Y^{5}}{175} - \frac{2007666 \sqrt{2} Y^{6}}{175} \\
\bottomrule
\end{tabular}
\end{center}
\needspace{6\baselineskip}
\end{samepage}
Beta$(2,2)$, high regime:
\begin{samepage}
\nopagebreak
\begin{center}\small
\renewcommand{\arraystretch}{1.35}
\begin{tabular}{>{$}l<{$}>{$}l<{$}}
\toprule
\varphi & \text{coefficient } P_\varphi(Y) \text{ in } F^{(2,2)},\ \tfrac14 \le Y \le \tfrac13 \\
\midrule
1 & \frac{4409}{37800} +\frac{9099 Y}{1400} - \frac{1701 Y^{2}}{20} \\
 & \quad +729 Y^{3} - \frac{26244 Y^{4}}{5} - \frac{1436859 Y^{5}}{175} \\
 & \quad +\frac{319302 Y^{6}}{175} \\
\pi & \frac{4917 \sqrt{2}}{5600} - \frac{1728 \sqrt{3} Y^{\frac{3}{2}}}{35} +\frac{93312 \sqrt{3} Y^{\frac{5}{2}}}{175} \\
 & \quad - \frac{93312 \sqrt{3} Y^{\frac{7}{2}}}{35} +\frac{39366 \sqrt{2} Y^{4}}{5} +20736 \sqrt{3} Y^{\frac{9}{2}} \\
 & \quad +\frac{9034497 \sqrt{2} Y^{5}}{350} - \frac{1003833 \sqrt{2} Y^{6}}{175} \\
\sqrt{9Y-2} & \frac{26629}{37800} +\frac{11551 Y}{4200} - \frac{1251 Y^{2}}{350} \\
 & \quad - \frac{51354 Y^{3}}{175} +\frac{1231767 Y^{4}}{175} \\
\arctan\sqrt2 & \frac{36693 \sqrt{2} Y}{400} - \frac{26973 \sqrt{2} Y^{2}}{40} +\frac{2673 \sqrt{2} Y^{3}}{2} \\
 & \quad - \frac{39366 \sqrt{2} Y^{4}}{5} - \frac{9034497 \sqrt{2} Y^{5}}{350} +\frac{1003833 \sqrt{2} Y^{6}}{175} \\
A_3 & - \frac{36693 \sqrt{2} Y}{400} +\frac{26973 \sqrt{2} Y^{2}}{40} - \frac{2673 \sqrt{2} Y^{3}}{2} \\
A_4 & \frac{6912 \sqrt{3} Y^{\frac{3}{2}}}{35} - \frac{373248 \sqrt{3} Y^{\frac{5}{2}}}{175} \\
W & \frac{39366 \sqrt{2} Y^{4}}{5} +\frac{9034497 \sqrt{2} Y^{5}}{350} - \frac{1003833 \sqrt{2} Y^{6}}{175} \\
B_+ & - \frac{4917 \sqrt{2}}{5600} +\frac{559872 \sqrt{3} Y^{\frac{7}{2}}}{175} - \frac{39366 \sqrt{2} Y^{4}}{5} \\
 & \quad - \frac{124416 \sqrt{3} Y^{\frac{9}{2}}}{5} - \frac{9034497 \sqrt{2} Y^{5}}{350} +\frac{1003833 \sqrt{2} Y^{6}}{175} \\
B_- & \frac{4917 \sqrt{2}}{5600} +\frac{559872 \sqrt{3} Y^{\frac{7}{2}}}{175} +\frac{39366 \sqrt{2} Y^{4}}{5} \\
 & \quad - \frac{124416 \sqrt{3} Y^{\frac{9}{2}}}{5} +\frac{9034497 \sqrt{2} Y^{5}}{350} - \frac{1003833 \sqrt{2} Y^{6}}{175} \\
\arctan(4\sqrt2-3\sqrt3) & - \frac{4917 \sqrt{2}}{5600} - \frac{559872 \sqrt{3} Y^{\frac{7}{2}}}{175} - \frac{39366 \sqrt{2} Y^{4}}{5} \\
 & \quad +\frac{124416 \sqrt{3} Y^{\frac{9}{2}}}{5} - \frac{9034497 \sqrt{2} Y^{5}}{350} +\frac{1003833 \sqrt{2} Y^{6}}{175} \\
\arctan(4\sqrt2+3\sqrt3) & - \frac{4917 \sqrt{2}}{5600} +\frac{559872 \sqrt{3} Y^{\frac{7}{2}}}{175} - \frac{39366 \sqrt{2} Y^{4}}{5} \\
 & \quad - \frac{124416 \sqrt{3} Y^{\frac{9}{2}}}{5} - \frac{9034497 \sqrt{2} Y^{5}}{350} +\frac{1003833 \sqrt{2} Y^{6}}{175} \\
\bottomrule
\end{tabular}
\end{center}
\end{samepage}

\section{Visual atlas of the spherical derivation}\label{app:atlas}

The derivation of Appendix~\ref{app:trig} is fully specified by its
formulas, but each object has a picture.  This appendix collects
them: the cap's azimuthal anatomy across the three upper-regime
phases (Figure~\ref{fig:atlas-caps}), the cell measure $m(\psi)$
against direct sampling (Figure~\ref{fig:atlas-mpsi}), and the
orientation views of the rhombic dodecahedron
(Figure~\ref{fig:atlas-angles}); the face-level picture
of the same mechanism appears as Figure~\ref{fig:atlas-sphere} in the main text.

\begin{figure}[p]
\centering
\includegraphics[width=\linewidth]{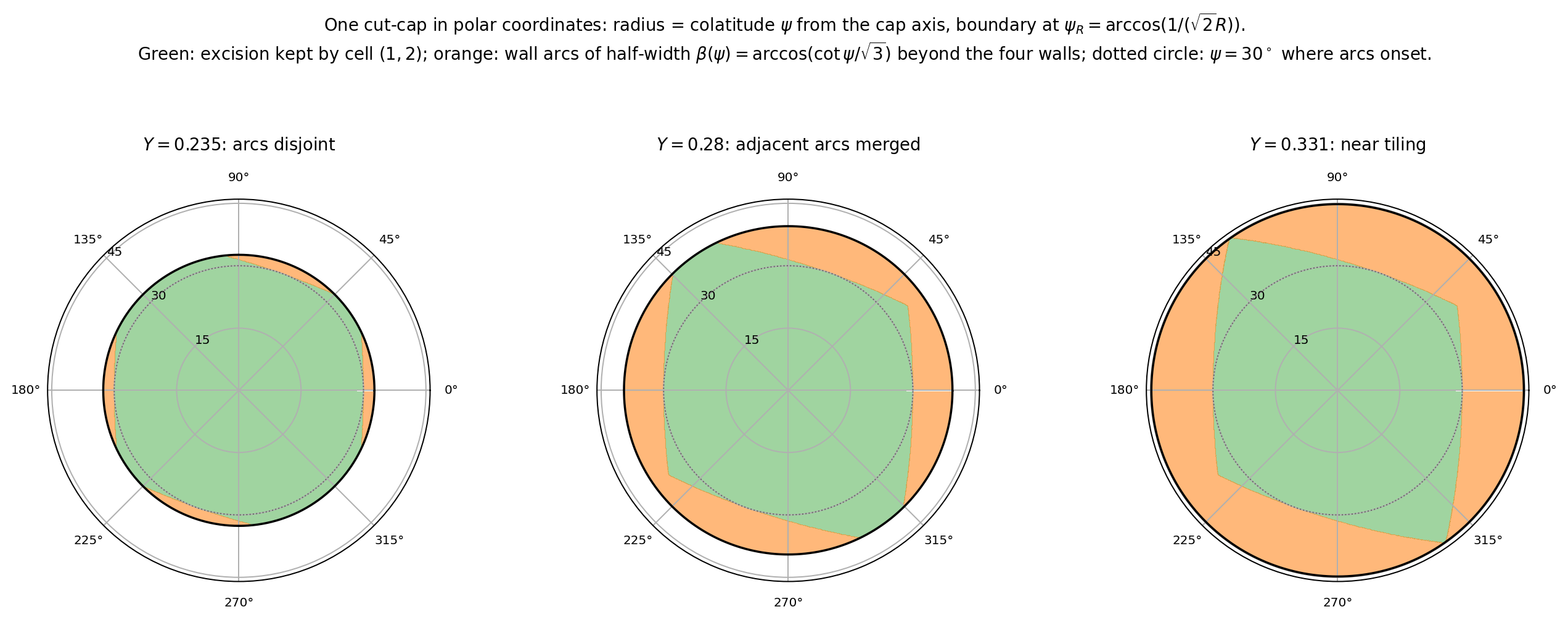}
\caption{Azimuthal anatomy of one cap through the upper regimes, in
polar coordinates with radius $=$ colatitude $\psi$ and the cap
boundary $\psi_R$ drawn bold.  Wall arcs of half-width
$\beta(\psi)$ open at $\psi = 30^\circ$ around the four wall
azimuths with alternating gaps \eqref{eq:app-gaps}.
\emph{Left} ($Y = 0.235$): the four arcs are disjoint---regime
\eqref{eq:fR4a}, second line of \eqref{eq:app-m}.
\emph{Middle} ($Y = 0.28 > 1/4$): the arcs across the two narrow
gaps have merged---regime \eqref{eq:fR4b}, third line of
\eqref{eq:app-m}.
\emph{Right} ($Y = 0.331$): approaching the ceiling the
cell-restricted cap grows while its surviving boundary intervals
shrink, toward the tiling identity \eqref{eq:app-tiling}, where the
twelve cell-restricted caps tile the sphere and the density
vanishes.}
\label{fig:atlas-caps}
\end{figure}

\begin{figure}[p]
\centering
\includegraphics[width=0.85\linewidth]{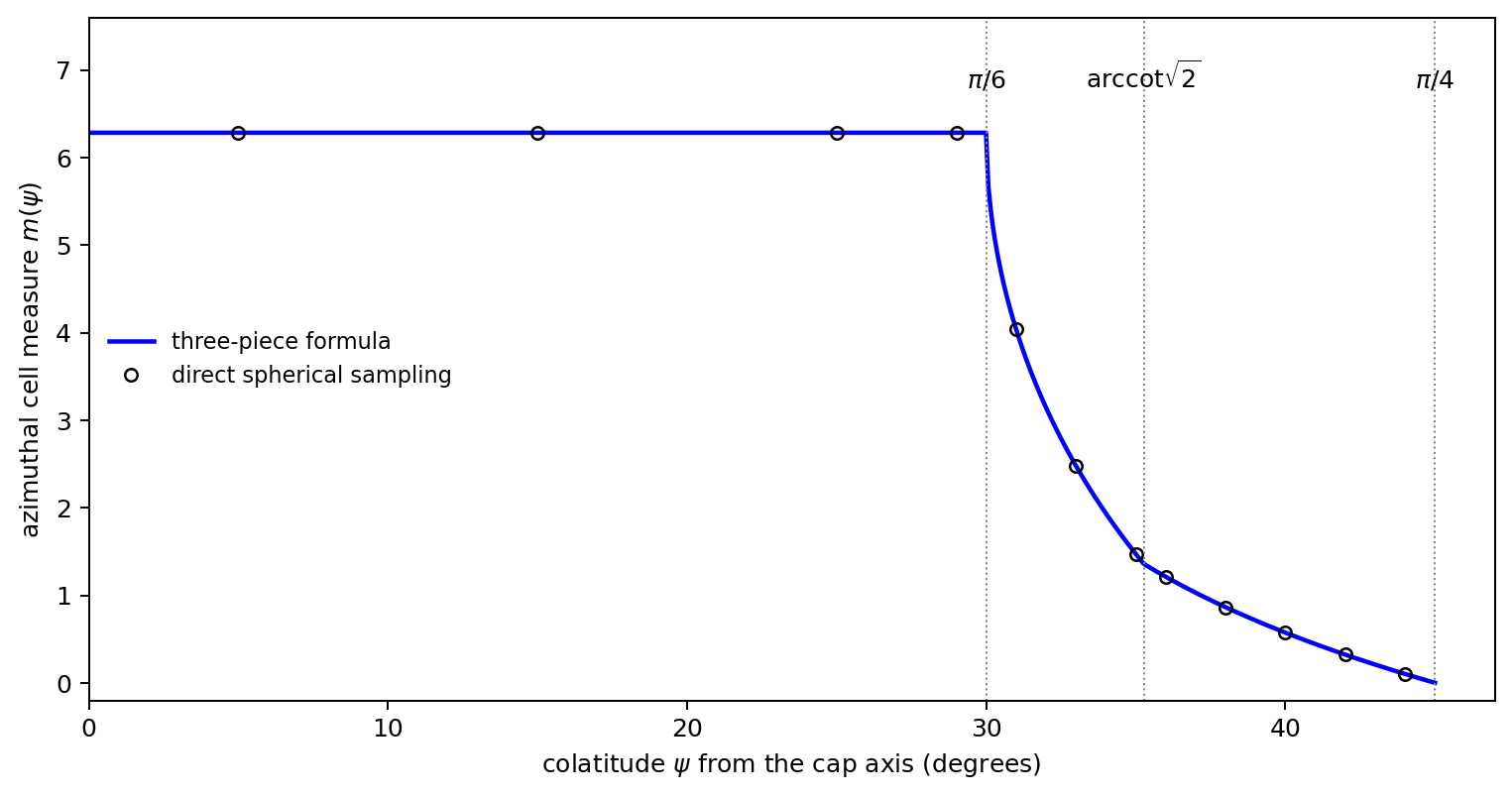}
\caption{The azimuthal cell measure $m(\psi)$ of \eqref{eq:app-m}:
the three-piece elementary formula against direct sampling of the
cell condition on the sphere.  The square-root onset at $\pi/6$ (an infinite right slope,
not a corner kink) is the arc onset,
the kink at $\operatorname{arccot}\sqrt2$ is the adjacent-arc merge,
and the zero at $\pi/4$ is the tiling identity; through
$\psi_R(Y)$ these become the CDF thresholds $2/9$, $1/4$, and
$1/3$.}
\label{fig:atlas-mpsi}
\end{figure}

\begin{figure}[p]
\centering
\includegraphics[width=\linewidth]{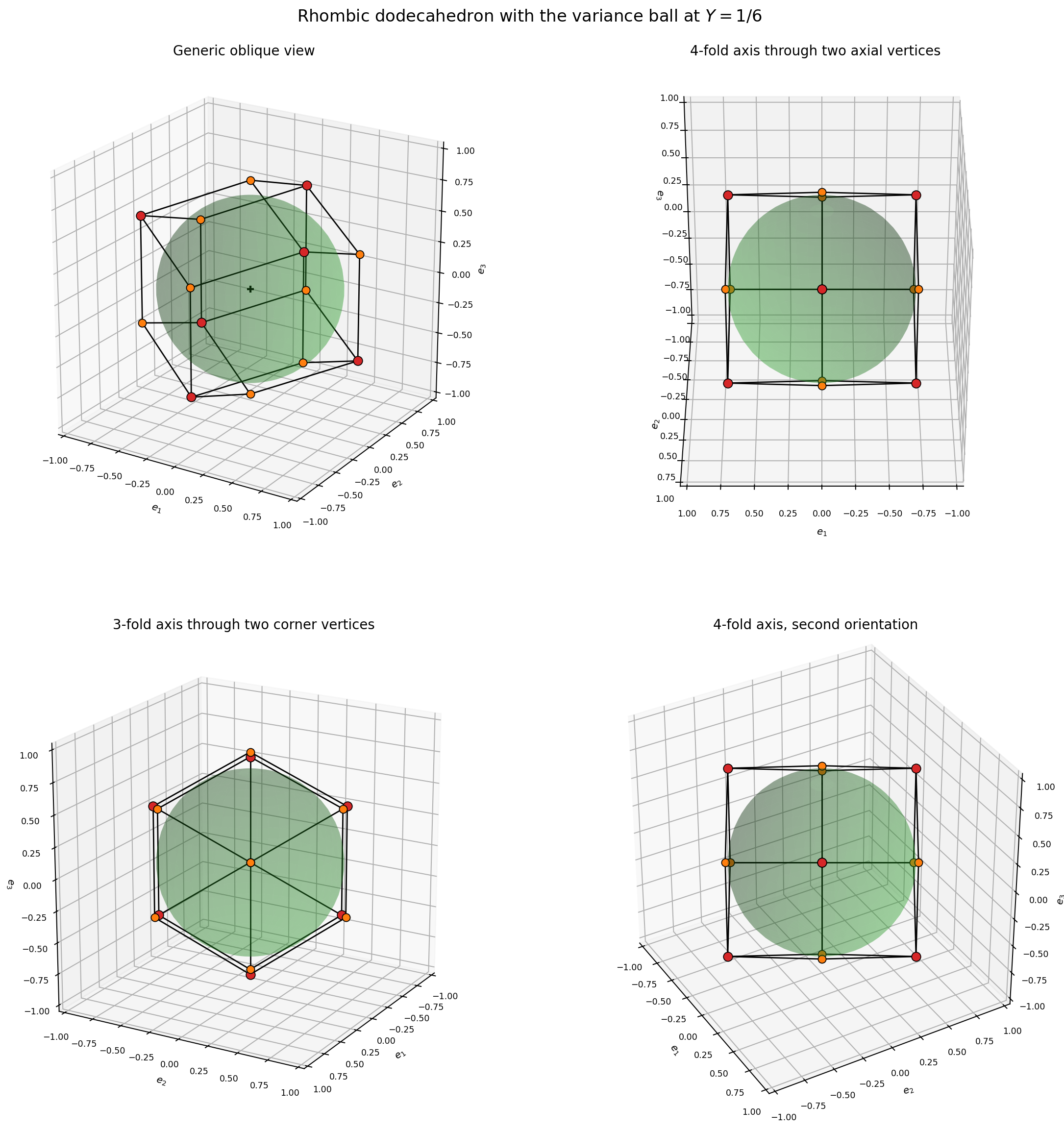}
\caption{Orientation views of the rhombic dodecahedron with the
variance ball at the face threshold $Y = 1/6$.  \emph{Top left:}
generic oblique view.  \emph{Top right and bottom right:} views
along two different $4$-fold axes, each through two antipodal axial
vertices, square silhouette.  \emph{Bottom left:} view along a $3$-fold axis through
two antipodal corner vertices, hexagonal silhouette---the direct
visual parallel of the $n = 3$ hexagon.  The symmetry visible here is the full octahedral group $O_h$ of the
polytope, whose subgroup $T_d \cong S_4$ of coordinate permutations
makes the twelve cap contributions of \eqref{eq:app-assembly} equal.}
\label{fig:atlas-angles}
\end{figure}

\end{document}